\documentclass[journal,twoside,web]{ieeecolor}
\usepackage{generic}
\usepackage{cite}
\usepackage{amsmath,amssymb,amsfonts}
\usepackage{algorithmic}
\usepackage{graphicx}
\usepackage{textcomp}
\usepackage{epstopdf}
\usepackage{balance}

\newcommand{\Hb}{\mathbf{H}}

\newcommand{\E}{{\mathbb E}}

\newcommand{\rank}{{\rm rank}}

\newcommand{\bmat}{\begin{bmatrix}}
\newcommand{\emat}{\end{bmatrix}}

\newcommand{\mR}{{\mathbb R}}

\newtheorem{lemma}{\bf Lemma}
\newtheorem{theorem}{\bf Theorem}

\newtheorem{corollary}{\bf Corollary}
\newtheorem{remark}{\bf Remark}

\newtheorem{computation}{\bf Procedure}

\def\BibTeX{{\rm B\kern-.05em{\sc i\kern-.025em b}\kern-.08em
    T\kern-.1667em\lower.7ex\hbox{E}\kern-.125emX}}
\title{Modeling of Low Rank Time Series}
\author{Wenqi Cao, \IEEEmembership{Student Member, IEEE}, Anders Lindquist, \IEEEmembership{Life Fellow, IEEE}\\ and Giorgio Picci, \IEEEmembership{Life Fellow, IEEE}
\thanks{Wenqi Cao is with Department of Automation, Shanghai Jiao Tong University, Shanghai, China. {\tt\small wenqicao@sjtu.edu.cn}}%
\thanks{Anders Lindquist is with the Department of Automation and the School of Mathematical Sciences, Shanghai
Jiao Tong University, Shanghai, China. {\tt\small alq@math.kth.se}}%
\thanks{Giorgio Picci is with the Department of Information Engineering, University of Padova, Italy. {\tt\small picci@dei.unipd.it}}%
}

\begin{document}

\maketitle

\begin{abstract}
 Rank-deficient stationary stochastic vector processes are present in many problems in network theory and dynamic factor analysis. In this paper we study hidden dynamical relations between the components of a discrete-time stochastic vector process and investigate their properties with respect to stability and causality. More specifically, we construct  transfer functions with a full-rank input process formed from selected components of the given vector process and having a vector process of the remaining components as output. An important question, which we answer in the negative, is whether it is always possible to find such a deterministic relation that is stable.  If it is unstable, there must be feedback from output to input ensuring that stationarity is maintained. This leads to connections to robust control.  We also show how our results could be used to investigate the structure of  dynamic network models and the latent low-rank stochastic process in a dynamic factor model.
\end{abstract}

\begin{IEEEkeywords}
 modeling, rank-deficient processes, dynamic factor models, stability, Granger causality, robust control
%Enter key words or phrases in alphabetical
%order, separated by commas. For a list of suggested keywords, send a blank
%e-mail to keywords@ieee.org or visit \underline
%{http://www.ieee.org/organizations/pubs/ani\_prod/keywrd98.txt}
\end{IEEEkeywords}

\section{Introduction}
\label{sec:introduction}

The basic topic of this paper is  modeling of discrete-time stochastic vector processes with rank-deficient spectral densities, i.e., reduced-rank (or low rank) processes. In  some literature such processes are also called sparse (or singular) signals. Processes of this kind are encountered in many practical applications where the data set is large with many correlated variables.
 In such applications available results for full-rank cases are no longer applicable, and therefore, as the theory of full-rank processes has become mature, many researchers have turned to rank-deficient processes in recent years, often in the form of different system realizations; see \cite{sglSS19} for applications in navigation and \cite{sglcov20,Ferrante-18letters}.

Rank-deficient processes may appear in networks, in which the nodes correspond to the components of the vector process. \cite{BCV18,WVD18,BGHP17,MS12,JPC19}.
Such dynamic network models are often needed in areas like econometrics, biology and engineering\cite{nweco10,nwbio09, nweng91}.
Rank-deficiency processes also play an important role in dynamic factor models \cite{Deistler19,AD12,EJC,ARYuleWalker} and generalized factor analysis (GFA) models \cite{BP15,Pms}, where there are latent processes with singular spectral densities. % so that the observations can be denoised.
Research on modeling and estimation of  models \cite{Deistler19,EJC} with singular latent processes often aim at representing the latent process as a product of a minimal common factor and some gains, called factor loadings \cite{BP15}. In econometrics singular AR and ARMA systems are important in the context of latent processes  modeling \cite{EJC,Deistler19,Ferrante-20cdc}, systems stability \cite{ARYuleWalker,AD12} and DSGE (dynamic stochastic general equilibrium) models.
Singular processes are also studied for state space models to extract the dynamical relations between the correlated entries \cite{GLsampling,CLPcdc20,PCLsysid21,BarnettSeth15} which relates to Granger causality \cite{Granger63,Granger69}.

Granger causality frequently appears in the context of feedback and was originally proposed for economics, then applied in control and information \cite{BarnettSeth15,Caines,PD12,ADD19} and neurophysiological systems \cite{BS17}, revealing the causal relations between the entries.
In addition, identification also needs to be considered when modeling the low rank processes, so that the approaches guarantee identifiability, rapidity, as well as the recovering of the real dynamical relations \cite{CLPcdc20,VARid}. Identification of low rank processes is considered in \cite{PCLsysid21,CPLauto21}.

Let $\{\zeta(t), t\in\mathbb{Z}\}$ be a stationary $p$-dimensional discrete-time stochastic vector process  with a rational spectral density $\Phi(e^{i\theta})$ of rank $m<p$. Then by rearranging the components of $\zeta$ if necessary, there is a decomposition
\begin{equation}
\label{eq:zeta2uy}
\zeta =\begin{bmatrix}u\\y\end{bmatrix}
\end{equation}
with $u$ an m-dimensional process with a spectral density $\Phi_u(e^{i\theta})$ of full rank $m$. As $u$ and $y$ are jointly stationary, there is a natural decomposition of $\Phi(z)$, namely
\begin{equation}
\label{Phidecomp}
\Phi(z)=\begin{bmatrix}
\Phi_u(z)&\Phi_{uy}(z)\\ \Phi_{yu}(z)&\Phi_y(z)
\end{bmatrix}
\end{equation}
where $\Phi_u(z)$ is full rank a.e.
An important question to be studied in this paper is whether there is a deterministic dynamic relation between $u$ and the $p-m$-dimensional vector process $y$. More precisely, with $z$ the time shift operator such that $zu(t)=u(t+1)$,  is there a description of the  relation between the input $u$ and  the output $y$ by  a   rational transfer function $F(z)$   as in Figure~\ref{system},  and, if so, how is it determined?
\begin{figure}[h]%[thpb]
      \centering
      \includegraphics[scale=0.12]{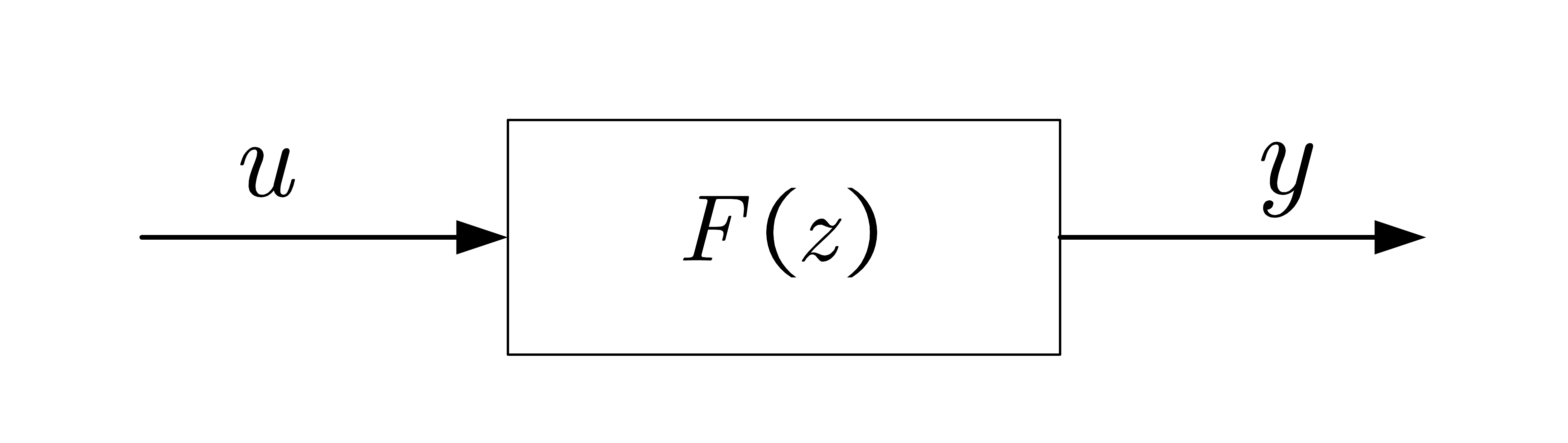}
      \caption{Dynamical relation from $u$ to $y$}
      \label{system}
\end{figure}
Such deterministic relations exist and were studied in \cite{GLsampling} in the context of sampling of continuous-time models, which leads to increase of rank.
In dynamic factor analysis $u$ can play the role of a minimal static factor \cite{Deistler19}. An important application where this problem occurs is a network where the components of $\zeta$ correspond to the nodes and one needs to find the transfer function between specified groups of nodes.

Although both $u$ and $y$ are stationary processes, $F(z)$ in Figure~\ref{system} need not be stable, as often erroneously assumed. In fact, $u$ is in general affected by $y$ and a key result of this paper, see Theorem \ref{MainThm}, is that the two signals are related by a feedback model of the following structure
\begin{subequations}
  \begin{align*}
    y &= F(z)u,\\
    u & = H(z)y + r ,
 \end{align*}
\end{subequations}
where $r$ is an external stationary input. We shall  prove this representation  in the context of general feedback models, in Section~\ref{sec:feedback} below. In such  a representation, $F(z)$ turns out to be  unique and   uniquely determined by the spectral density decomposition \eqref{Phidecomp} as per  formula \eqref{Fred} to be derived later on.

In general the components of $\zeta$ can be rearranged in several different ways, thus producing different $u$ and $y$ in \eqref{eq:zeta2uy} while upholding the requirement that $u$ has full rank $m$.  An important question is whether there is always a choice for which $F(z)$ is a stability matrix. In this paper we shall demonstrate by way of counterexample that this is not always possible.

 In Section~\ref{sec:Fconstruction} we shall also provide a state-space  algorithm for determining  it. However $H(z)$ is not unique, but, since $\zeta(t)$ is stationary, it has to be chosen so that the feedback loop is internally stable \cite{DFT}. This leads to an interesting connection to  robust control, which  is the topic of Section~\ref{sec:H} .

The outline of the paper is as follows.  Section~\ref{sec:feedback} will be devoted to feedback models, inserting the deterministic dynamic system in Figure~\ref{system} in a stochastic feedback environment.  In Section~\ref{sec:H} we consider the design of the feedback environment needed for internal stability, thus connecting to robust control \cite{DFT}. In Section~\ref{sec:Fconstruction} we shall present a realization for $F(z)$ of degree at most $n$, and investigate its properties. Section~\ref{sec:stabilitycausality} deals with stability and causality.  In Section~\ref{sec:networks}  we explore the connections to dynamic network models, which is an important application. In Section~\ref{sec:factormodels} we show how our results could be used to investigate the structure of the latent low-rank stochastic process in a dynamic factor model. Section~\ref{sec:examples}  provides some further examples to illustrate our results. Finally, in Section~\ref{sec:conclusion} we give some conclusions.

\section{Feedback representations}\label{sec:feedback}

The two processes $u$ and $y$ in \eqref{eq:zeta2uy} are jointly stationary, and we can express both $y(t)$ and $u(t)$ as a sum of the best linear estimate based on the past of the other process plus an error term, i.e.,
 \begin{subequations}\label{yu}
           \begin{align}
   y(t) &= \mathbb{E}\{y(t)\mid \mathbf{H}_t^-(u)\} + v(t),\label{u2y}\\
    u(t) &= \mathbb{E}\{u(t)\mid \mathbf{H}^-_{t}(y)\} + r(t)\label{y2u}
\end{align}
\end{subequations}
where $\mathbf{H}_t^-(u)$ is the closed span of the past components $\{ u_1(\tau), u_2(\tau), \dots, u_q(\tau)\}\mid \tau \leq t\}$ of the vector process $u$ in the Hilbert space of random variables, and let $\mathbf{H}_t^-(y)$ be defined likewise in terms of $\{ y_1(\tau), y_2(\tau), \dots, y_p(\tau)\mid \tau \leq t\}$. For future use, we shall also need the closed span $\Hb_t^+(u)$ of the future components $\{ u_1(\tau), u_2(\tau), \dots, u_q(\tau)\}\mid \tau \geq t\}$  and the closed span $\Hb(u)$ of the complete (past and future) history of $u$, and similarly for $y$.
Each linear projection in \eqref{yu} can be represented by a linear filter which we write as
\begin{subequations}\label{FuHy}
           \begin{align}
    y &= F(z)u + v,\\
    u & = H(z)y + r ,
 \end{align}
  \end{subequations}
 where $F(z)$ and $H(z)$ are proper rational transfer functions of dimensions $(p-m)\times m$ and $m\times (p-m)$, respectively. Hence we have the feedback configuration depicted in Figure~\ref{fbdtfig}.
\begin{figure}[h]%[thpb]
      \centering
      \includegraphics[scale=0.45]{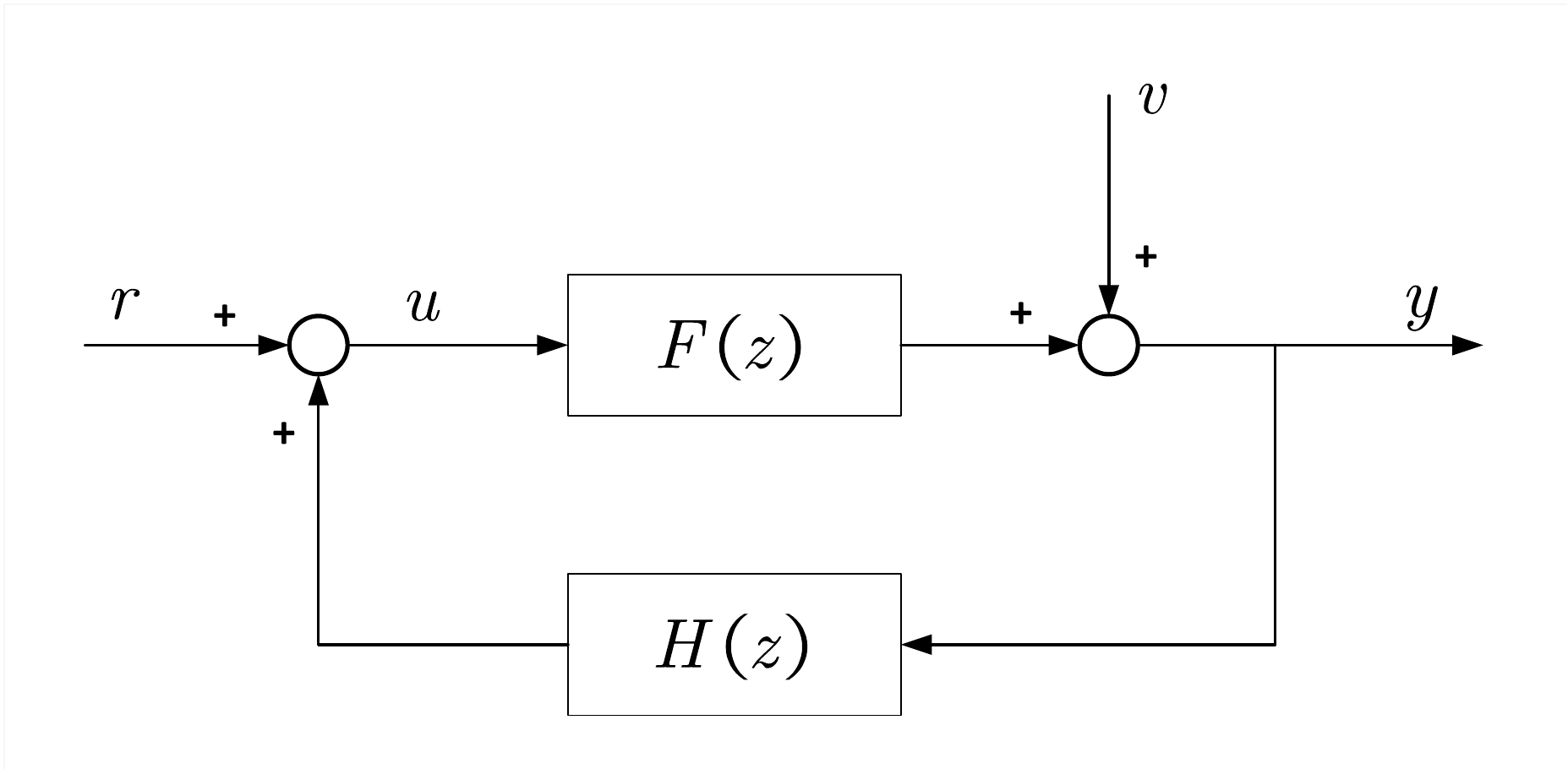}
      \caption{Block diagram illustrating  a feedback representation}
      \label{fbdtfig}
\end{figure}

The transfer functions  $F(z)$ and $H(z)$ are in general not stable, but, since  all processes are jointly stationary, the overall feedback configuration needs to be internally stable \cite{DFT,KZhou}. The error processes $v$ and  $r$ produced by \eqref{yu} are in general correlated. However,  it   will be shortly recalled in Appendix \ref{Uncorrvr} below, that there are always   feedback models \eqref{FuHy} where  $v$ and  $r$ are uncorrelated, and in the sequel we shall assume that this is so. In this context, we shall revisit some results in \cite{Caines}, where however a full-rank requirement to insure uniqueness was imposed.  After that,   we shall  allow for rank-deficient spectral densities.

\medskip

 \begin{lemma}\label{lem1}
        The transfer function matrix $T(z)$ from $\begin{bmatrix}r\\v\end{bmatrix}$ to $\begin{bmatrix}u\\y\end{bmatrix}$ in the feedback model \eqref{FuHy} is given by
\begin{subequations}\label{TPQ}
\begin{equation}
\label{T}
T(z)=\begin{bmatrix}P(z)&P(z)H(z)\\Q(z)F(z)&Q(z)\end{bmatrix},
\end{equation}
where
\begin{equation}
  \label{PQ}
\begin{split}
P(z)&= (I - H(z)F(z))^{-1},\\
Q(z)&=(I - F(z)H(z))^{-1}
\end{split}
\end{equation}
\end{subequations}
are strictly stable.
 Moreover,
 \begin{equation}
\label{PQFH}
Q(z)F(z)=F(z)P(z), \quad H(z)Q(z)=P(z)H(z),
\end{equation}
 and $T(z)$  is full rank and strictly stable.
 \end{lemma}

  \medskip

\begin{proof}
Since the stationary processes $v$ and $r$ produce stationary processes $y$ and $u$, the feedback model must be internally stable, and therefore $T(z)$ is (strictly) stable.
From \eqref{FuHy} we have
\begin{equation}
\label{rv2uy}
 \begin{bmatrix}u\\ y\end{bmatrix}=\begin{bmatrix}0 &H(z)\\F(z) & 0\end{bmatrix}\begin{bmatrix}u\\ y\end{bmatrix}+ \begin{bmatrix}r\\ v\end{bmatrix}
\end{equation}
         and therefore
\begin{displaymath}
R(z)\begin{bmatrix}u\\ y\end{bmatrix}=\begin{bmatrix}r\\ v\end{bmatrix} ,
\end{displaymath}
               where
 \begin{displaymath}
 R(z) :=\begin{bmatrix}I & -H(z)\\-F(z) & I\end{bmatrix}
 \end{displaymath}
is stable and full rank, being the transfer function from $u,y$ to $r,v$. Then the Schur complements  $(I - H(z)F(z))$ and $(I - F(z)H(z))$ are also full rank, and hence we can form the transfer functions $P(z)$ and $Q(z)$ as in \eqref{PQ}. Moreover,
\begin{displaymath}
\det P(z)= \det R(z) =\det Q(z),
\end{displaymath}
and hence $P(z)$ and $Q(z)$ are also strictly stable.
A straightforward calculation shows that $T(z)R(z)=I$, and hence
$T(z) = R(z)^{-1}$, as claimed. It is immediately seen that $FP^{-1}=Q^{-1}F$ and that $P^{-1}H=HQ^{-1}$, and therefore \eqref{PQFH} follows.
 \end{proof}

 \medskip

\begin{remark}\label{rem:sensitivity}
Note that $Q(z)$ is the output sensitivity function -- generally merely called the sensitivity function -- and $P(z)$ the input sensitivity function of the feedback system in Figure~\ref{fbdtfig}, and hence they must be strictly stable for this reason also \cite[p.82]{KZhou},  \cite{DFT}.
\end{remark}

 \medskip

 The $p\times p$ spectral density $\Phi(e^{i\theta})$ of rank $m$ has a $p\times m$ stable spectral factor
\begin{equation}
\label{W}
W(z)=C(zI-A)^{-1}B +D,
\end{equation}
such that
\begin{equation}
\label{eq:Phi}
W(z)W(z^{-1})'=\Phi(z).
\end{equation}
 Note that we do not require \eqref{W} to be minimum phase. In fact, we shall also cover the case when $W(z)$ has a zero at infinity.
Then $\zeta$ has a minimal stochastic realization
\begin{subequations}\label{eq:system}
\begin{align}
\label{eq:x}
x(t+1)&=Ax(t)+ Bw(t)\\
\label{eq:zeta}
\zeta(t)&=Cx(t) +Dw(t)
\end{align}
\end{subequations}
of dimension $n$ with $C\in \mR^{p\times n}$, $A\in \mR^{n\times n}$, $B\in \mR^{n\times m}$, $D\in \mR^{p\times m}$ and $\{w(t), t\in\mathbb{Z}\}$
an $m$-dimensional normalized white noise process such that  $\mathbb{E}\{w(t)w(t)'\} =I\delta_{ts}$. Moreover,  $A$ is a stability matrix, i.e., having all eigenvalues in the unit disc, and $(C,A)$ and $(A,B)$ are  observable and reachable pairs, respectively \cite{LPbook}.

 The white noise process $w$ in \eqref{eq:system} has the spectral representation
 \begin{displaymath}
w(t)=\int_{-\pi}^\pi e^{it\theta}d\hat{w} \quad \text{where}\quad \mathbb{E}\{d\hat{w}d\hat{w}^*\}= \frac{d\theta}{2\pi}I
\end{displaymath}
\cite{LPbook}.  Then
\begin{equation}\label{WNrepr}
 \begin{bmatrix}u(t)\\ y(t)\end{bmatrix} =\int_{-\pi}^\pi e^{it\theta}W(e^{i\theta})d\hat{w}.
\end{equation}
 Next since $v$ and $r$ are assumed uncorrelated, they can be represented as \begin{equation}
\label{eq:w2rv}
\begin{bmatrix}r(t)\\ v(t)\end{bmatrix}= \int_{-\pi}^\pi e^{it\theta}\begin{bmatrix}K(e^{i\theta})&0\\0&G(e^{i\theta})\end{bmatrix}d\hat{w},
\end{equation}
where the transfer functions  $K(z)$ of dimension $m\times m_1$ and $G(z)$ of dimension $(p-m)\times m_2$ for some $m_1\leq m$ and $m_2\leq (p-m)$ such that $m_1+m_2=m$. can be chosen minimum phase.   Moreover, it follows from Lemma~\ref{lem1} that
 \begin{displaymath}
W= \begin{bmatrix}P&PH\\QF&Q\end{bmatrix}\begin{bmatrix}K&0\\0&G\end{bmatrix}=  \begin{bmatrix}PK&PHG\\QFK&QG\end{bmatrix},
\end{displaymath}
which in view of \eqref{PQFH}  can be written
\begin{equation}
\label{Wdecomp1}
W= \begin{bmatrix}W_{11}&W_{12}\\W_{21}&W_{22}\end{bmatrix}=\begin{bmatrix}PK&HQG\\FPK&QG\end{bmatrix},
\end{equation}
where the subscripts refer to the partitioning of $W$ induced by the subdivision \eqref{Phidecomp}.
Consequently,
\begin{equation}
\label{eq:W2FH}
W_{21}=FW_{11}\quad\text{and}\quad W_{12}=HW_{22},
\end{equation}
so, if $W$ is full rank and thus $m=p$ as in \cite{Caines}, then $F=W_{21}W_{11}^{-1}$ and $H=W_{12}W_{22}^{-1}$. However, in our setting $m<p$ and hence $F$ and $H$ cannot both be uniquely determined from \eqref{eq:W2FH}.

By  \eqref{eq:w2rv} and  Lemma~\ref{lem1}, we have
\begin{equation}
\label{Phivr}
\Phi(z)=T(z) \begin{bmatrix}\Phi_r(z)&0\\0&\Phi_v(z)\end{bmatrix}T(z)^*,
\end{equation}
where $\Phi_v(z)=G(z)G(z)^*$ and $\Phi_r(z)=K(z)K(z)^*$ are  the spectral densities of $v$ and $r$, respectively, and $^*$ denotes transpose conjugate. Since $T$ has full rank a.e. (almost everywhere), $\Phi$ is rank deficient if and only if at least one of $\Phi_v$ or $\Phi_r$ is.

 \medskip

 \begin{theorem}\label{lem2}
Suppose $(H\Phi_vH^*+\Phi_r)$ is positive definite a.e. on unit circle. Then
\begin{equation}
\label{F}
F=\Phi_{yu}\Phi_u^{-1} - \Phi_vH^*(H\Phi_vH^*+\Phi_r)^{-1}(I-HF),
\end{equation}
that is
\begin{equation}
\label{Fred}
F=\Phi_{yu}\Phi_u^{-1}
\end{equation}
if and only if $\Phi_v H^*\equiv 0$.
\end{theorem}

\medskip

\begin{proof}
Given \eqref{Phidecomp}, \eqref{TPQ} and \eqref{Phivr}, we have
\begin{align*}
 \Phi_u  &=  P(H\Phi_vH^*+\Phi_r)P^*=HQ\Phi_vH^*P^*+P\Phi_rP^* \\
   \Phi_{yu}&= Q(\Phi_vH^* +F\Phi_r)P^*=Q\Phi_vH^*P^* + FP\Phi_rP^*,
\end{align*}
where we have used \eqref{PQFH}, i.e., $QF=FP$ and $HQ=PH$.
Hence, in view of \eqref{PQFH},
\begin{displaymath}
 \Phi_{yu}-F \Phi_u  = (I-FH)Q\Phi_vH^*P^*=\Phi_vH^*P^*,
\end{displaymath}
which yields
\begin{displaymath}
 \Phi_{yu} \Phi_u ^{-1}-F=\Phi_vH^*(H\Phi_vH^*+\Phi_r)^{-1}P^{-1},
\end{displaymath}
from which \eqref{F} follows by \eqref{PQ}.
\end{proof}

\medskip

Next, we specialize to feedback models of rank deficient processes. We shall show that there are  feedback model representations  where the forward   channel is described by  a {\em deterministic relation} between $u$ and $y$.

\begin{theorem}\label{MainThm}
Let $\zeta$ be a $p$-dimensional process  of rank $m$. Any partition \eqref{eq:zeta2uy} with a full rank $m$-dimensional subvector process $u$ of $\zeta$ can be represented by an internally stable feedback model of the form
\begin{subequations}\label{FuHyv=0}
 \begin{align}
        y &= F(z)u \\
         u &= H(z)y + r,
\end{align}
\end{subequations}
   where the transfer functions $F(z)$  is  uniquely determined in terms the joint spectra of the two components $u$ and $y$ by formula \eqref{Fred} and the input noise $r$ is a stationary process of full rank $m$.
\end{theorem}
\begin{proof} See Appendix \ref{appendixThm_2}.
\end{proof}

Hence  for rank-deficient processes there is a  fixed  deterministic dynamical relation from $u$ to $y$ as depicted in Figure~\ref{fig_specialFB}, where $F(z)$ is fixed, given by \eqref{Fred}. Note that  a nontrivial $H(z)$ will permit $F(z)$ to be unstable, as the feedback should  stabilize the feedback loop.
 \begin{figure}[h]%[thpb]
      \centering
      \includegraphics[scale=0.4]{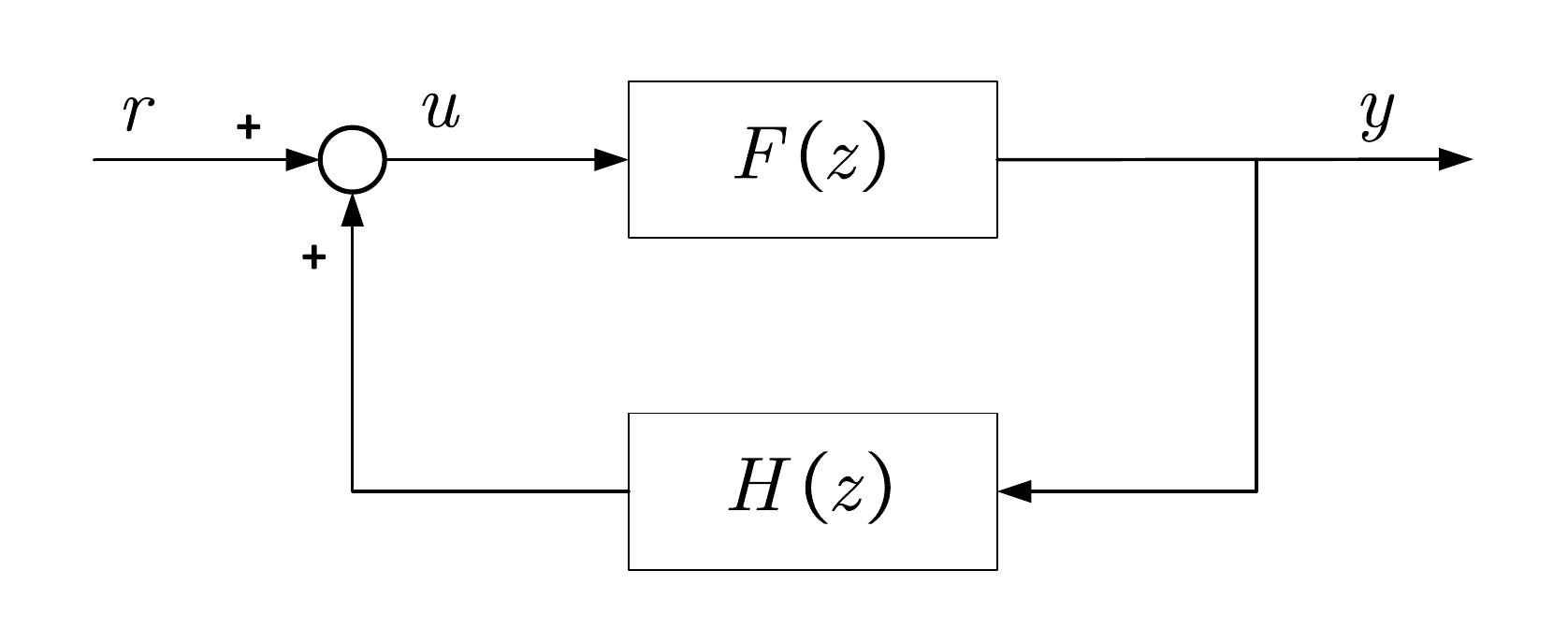}
      \caption{Feedback representation of rank-deficient processes}
      \label{fig_specialFB}
\end{figure}
\medskip

Manfred Deistler \cite{Deistler} has recently posed the question whether there is always a   selection of the  components in $\zeta$ admitting a stable $F(z)$. We shall answer this question with a counterexample  in the negative later, in section \ref{sec:stabilitycausality}. However before doing this we shall need to study the question of feedback stabilization of the system \eqref{FuHyv=0} in general and successively use  state space representations to discuss the question in depth.

\section{How to determine $H(z)$}\label{sec:H}
Given the transfer function $F(z)$ in \eqref{Fred}, determining the corresponding $H(z)$ in \eqref{FuHyv=0}  is a feedback stabilization problem studied in robust control \cite{DFT,KZhou}, which may have several solutions.  In fact, regarding $F(z)$ in  Fig.~\ref{fig_specialFB} as a plant, we need to determine a compensator $H(z)$ so that the feedback system is internally stable. As pointed out in Remark~\ref{rem:sensitivity}, the sensitivity function is
\begin{displaymath}
Q(z)=[I-F(z)H(z)]^{-1},
\end{displaymath}
i.e., the function $Q(z)$ in \eqref{PQ}. For the feedback system to be internally stable, $Q(z)$ needs to be analytic in the complement of the open  unit disc (i.e., strictly stable) and satisfy certain interpolation conditions in unstable poles and non-minimum-phase zeros of $F(z)$ \cite{DFT}. In addition, in robust control design, there must be a bound
\begin{displaymath}
\|Q\|_\infty \leq \gamma.
\end{displaymath}
There is a minimum bound $\gamma_\text{opt}$, but we shall just choose $\gamma > \gamma_\text{opt}$. Finally $Q(z)$ should be rational of small McMillan degree. This is an analytic  (Nevanlinna-Pick) interpolation problem with rationality constraint \cite{BGL1,BGLcdc99,GL2,BLN}.

Given a solution $Q(z)$ of this analytic interpolation problem, we can determine $H(z)$ from $F(z)H(z)=I - Q(z)^{-1}$ provided that the left pseudo-inverse $F^\dag:=(F^*F)^{-1}F^*$ exists. If $F(z)$ is a long rectangular matrix function and the pseudo-inverse $F^\dag:=F^*(FF^*)^{-1}$ exists, we may instead formulate an analytic interpolation problem for the input sensitivity function
\begin{displaymath}
P(z)=(I-H(z)F(z))^{-1},
\end{displaymath}
(Remark~\ref{rem:sensitivity}) and solve for $H(z)$ from  $H(z)F(z)=I-P(z)^{-1}$.

To explain the basic ideas of the procedure in simple terms we shall first consider the case that $F(z)$ and $H(z)$, and thus also $Q(z)$, are scalar, in which case the interpolation conditions are simple. Then $Q(z)$ must send the unstable poles of $F(z)$ to $0$ and the non-minimum phase zeros to $1$ \cite{DFT}. Moreover, the function $f(z):=\gamma^{-1}Q(z^{-1})$ is a Schur function, i.e., a function that is analytic in the open unit disc and maps it into the open unit disc (Figure~\ref{fig:Schur}).
\begin{figure}[h]%[thpb]
      \centering
      \includegraphics[scale=0.5]{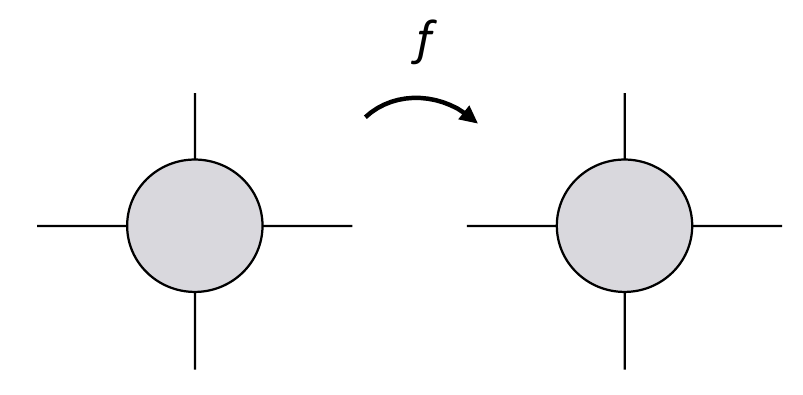}
      \caption{Schur function}
      \label{fig:Schur}
\end{figure}
Then, with the interpolation points $z_0,z_1\dots, z_\nu$,  the solutions of the analytic interpolation problem are completely parametrized by the polynomials in the class $\mathcal{S}$ of stable monic polynomials of degree $\nu$. More precisely, to each $\sigma\in\mathcal{S}$, there is a unique pair of polynomials $(\alpha, \beta)$ with $\alpha\in\mathcal{S}$ and $\beta$ a polynomial of at most degree $\nu$ such that $f(z)=\beta(z)/\alpha(z)$ satisfies the interpolation conditions and $|\alpha|^2-|\beta|^2=|\sigma|^2$ \cite{BGLM}.
 The solution corresponding to $\sigma$ is obtained by maximizing
\begin{equation}
\label{primal}
\int_{-\pi}^\pi \left|\frac{\sigma(e^{i\theta})}{\tau(e^{i\theta})}\right|^2 \log(1-|f(e^{i\theta})|^2)d\theta,
\end{equation}
 where $\tau(z)=\Pi_{k=0}^{\nu-1}(1-\bar{z}_kz)$, subject the interpolation conditions.\\
Choosing $\sigma(z)=z^n\tau(z^{-1})$, we obtain the central or maximum-entropy solution maximizing
 \begin{displaymath}
\int_{-\pi}^\pi  \log(1-|f(e^{i\theta})|^2)d\theta.
\end{displaymath}
which is a linear problem \cite{MustafaGlover}.

As a simple example,  let us consider the transfer function
 \begin{equation*}
\label{unstableF}
F(z) = \frac{26z-27}{2(2z-7)},
\end{equation*}
 which has one unstable pole at $z=7/2$ and one non-minimum phase zero at $z=27/26$, yielding the interpolation conditions $f(\xi_0)=0$ and $f(\xi_1)=\gamma^{-1}$, where $\xi_0=2/7$ and $\xi_1=26/27$. We may simplify the problem by moving the interpolation point $\xi_0$ to $z_0=0$, which can be done by the transformation
\begin{displaymath}
z=\frac{\xi -\xi_0}{1-\xi_0\xi}
\end{displaymath}
(for a real $\xi_0$), thus moving  $\xi_1$ to $z_1=(\xi_1-\xi_0)(1-\xi_0\xi_1)^{-1}=128/137$.

Alternatively, we may solve an analytic interpolation for a Carth{\'e}odory function
\begin{equation}
\label{phi}
\varphi(z)=\tfrac12\frac{1-f(z)}{1+f(z)},
\end{equation}
mapping the open unit disc to the open right half plane, yielding the interpolation conditions $\varphi(z_0)=\tfrac12$ and $\varphi(z_1)=\tfrac12(\gamma-1)(\gamma+1)^{-1}$. Then we can use the optimization procedures in \cite{BGL1,BGLcdc99} to determine an appropriate $Q(z)$.
%We can also apply the Riccati-type approach in \cite{CLtac} and obtain a solution  by the homotopy continuation method in subsection III-E of \cite{CLtac}.
However it is simpler to apply the Riccati approach in \cite{CLtac}, for which we have made the needed simple calculations in Appendix B.
%The parameter $\gamma$ has to be chosen so that the Pick condition in \cite[Proposition 3]{CLtac} is satisfied.}
Choosing the central solution, we obtain from \eqref{phicentral}, \eqref{phi} and \eqref{uU}
\begin{displaymath}
f(z)=\frac{1-2\varphi(z)}{1+2\varphi(z)}=-uz =\gamma^{-1}z_1^{-1}z,%=\frac{137}{128}\gamma^{-1}z,
\end{displaymath}
and hence, moving $z$ back to $\xi$,
%\begin{displaymath}
%Q(z)=\gamma f(z)=\xi_1^{-1}z,
%\end{displaymath}
\begin{displaymath}
  f(\xi)=f(z)\big|_{z=\frac{\xi -\xi_0}{1-\xi_0\xi}}=\gamma^{-1}z_1^{-1}\frac{\xi -\xi_0}{1-\xi_0\xi},
\end{displaymath}
then we have
\begin{displaymath}
  Q(z)=\gamma f(\xi)\big|_{\xi=z^{-1}}=z_1^{-1}\frac{1 -\xi_0z}{z-\xi_0}=\frac{137}{128}\frac{7 -2z}{7z-2},
\end{displaymath}
which is analytic outside the closed unit circle. It is easy to check that $Q$ satisfies $Q(7/2)=0$ and $Q(27/26)=1$, meaning that $Q(z)$ sends the unstable pole of $F(z)$ to $0$ and the non-minimum phase zero to $1$.
So,
\begin{displaymath}
\begin{split}
H(z)&=F(z)^{-1}(1-Q(z)^{-1})\\
&= \tfrac{2}{13}\frac{z-\xi_0^{-1}}{z-\xi_1^{-1}}(1-\frac{\xi_1-\xi_0}{1-\xi_0\xi_1}\frac{z-\xi_0}{1 -\xi_0z}),\\
&=\tfrac{2}{13}\frac{1-\xi_0^2}{1-\xi_0\xi_1}\frac{\xi_1}{\xi_0},
\end{split}
\end{displaymath}
that is
%\begin{displaymath}
%H(z)=\tfrac{13}{2}\frac{z-\xi_0}{z} =\tfrac{13}{2}\frac{z-\tfrac{2}{7}}{z},
%\end{displaymath}
\begin{displaymath}
  H(z)=\tfrac{90}{137},
\end{displaymath}
which is also stable.

The case when $Q(z)$ is matrix-valued is considerably more complicated, and we explain this in Section~\ref{Example4} in the context of a simple example, which we solve with the matrix version of the Riccati-type nonlinear  matrix equation \cite{CLtac}.
Moreover, we refer the reader to the literature, especially \cite{BLN,FPZbyrneslindquist,RFP,ZhuBaggio}. For the central solution, also see \cite{MustafaGlover}.

\section{Determining $F(z)$ from a spectral factor}\label{sec:Fconstruction}
Let
\begin{displaymath}
W(z)=C(zI-A)^{-1}B +D
\end{displaymath}
be  a minimal realization of dimension $n$ of a $p\times m$ spectral factor \eqref{W} of the system  \eqref{W}. We want to determine the unique $F(z)$ from $A,B,C$ and $D$.

 Following \cite{GLsampling}, in \cite{CLPcdc20} we gave a procedure to solve the problem in the continuous-time case. However,  then $D=0$, which considerably simplifies the situation. The discrete-time case requires to consider situations when $D$ is nonzero and with a rank $\rho\leq m$, which complicates the calculations considerably. We shall present a procedure for determining $F$ below; see Procedure 1.  To this end, we first need to develop the appropriate equations.

 Let $\rho$ be the rank of the $p\times m$ matrix $D$. To simplify calculations we shall first perform a singular value decomposition
\begin{subequations}\label{D2Sigma}
\begin{equation}
UDV'=\begin{bmatrix}\Sigma&0\\0&0\end{bmatrix},
\end{equation}
where $\Sigma$ is a diagonal $\rho\times \rho$ matrix consisting of the nonzero singular values, and $U$ and $V$ are orthogonal matrices of dimensions $p\times p$ and $m\times m$, respectively, i.e., $U'U=V'V=I$.  We assume that the corresponding transformations
\begin{equation}
\label{ }
(\zeta,w)\to (U\zeta,Vw)\quad\text{and}\quad (B,C)\to (BV',UC)
\end{equation}
have already been performed in \eqref{eq:system}. Moreover,
\begin{equation}
\Phi(z)\to U\Phi(z)U', \qquad W(z) \to UW(z)V'.
\end{equation}
\end{subequations}

Next partition the new matrices $C$ and $B$ as
\begin{equation}
\label{CBpartitioning}
C=\begin{bmatrix}C_0\\C_1\\C_2\end{bmatrix}
\qquad B=\begin{bmatrix}B_0&B_1\end{bmatrix},
\end{equation}
where $C_0$ $C_1$, $C_2$, $B_0$ and $B_1$ are $\rho\times n$, $(m-\rho)\times n$, $(p-m)\times n$, $n\times \rho$ and $n\times (m-\rho)$, respectively, after having changed, if necessary, the order of the component in $\zeta$ so that the square $(m-\rho)\times (m-\rho)$ matrix $C_1B_1$ is full rank. As we shall see below, this can always be done and in general in several different ways. Then the partitioning of $C$ leads to the representation
\begin{equation}
\label{eq:zetauy}
\zeta =\begin{bmatrix}u\\y\end{bmatrix}
\quad\text{where}\quad
u=\begin{bmatrix}u_0\\u_1\end{bmatrix}
\end{equation}
and the partitioning
\begin{equation}
\label{Wpart}
\begin{bmatrix}W_{00}&W_{01}\\W_{10}&W_{11}\\W_{20}&W_{21}\end{bmatrix}
\end{equation}
of the spectral factor \eqref{W}. (Note that this does not correspond to the decomposition \eqref{Wdecomp1}.) Consequently,
\begin{equation}
\label{W00}
W_{00}(z)=C_0(zI-A)^{-1}B_0 +\Sigma
\end{equation}
and
\begin{equation}
\label{Wjk}
W_{jk}(z)=C_j(zI-A)^{-1}B_k
\end{equation}
when $(j,k)\ne (0,0)$. Moreover, using Lemma~\ref{inverse} in the  appendix, we have
\begin{equation}
\label{W00inv}
W_{00}(z)^{-1} =\Sigma^{-1}\left[ I-C_0(zI-\Gamma_0)^{-1}B_0\Sigma^{-1}\right],
\end{equation}
where
\begin{equation}
\label{eq:Gamma0}
\Gamma_0=A-B_0\Sigma^{-1}C_0.
\end{equation}

\medskip

\begin{theorem}\label{thm:ufullrank}
Suppose that $C_1B_1$ is nonsingular. Then the $m$-dimensional process $u$ in \eqref{eq:zetauy} is full rank.
\end{theorem}

\medskip

\begin{proof}
See  Appendix~\ref{appendixthm3}
\end{proof}

\medskip

\begin{theorem}\label{thm:F(z)}
Suppose that the order of the components in $\zeta$ is chosen so that $C_1B_1$ is nonsingular.
Then the  transfer function $F(z)$ mapping $u$ to $y$ is given by
\begin{equation}
\label{F(z)}
F(z)=\left[F_0(z),F_1(z)\right],
\end{equation}
where
\begin{subequations}\label{F0F1}
\begin{align}
F_0(z)&= C_2(zI-\Gamma_1)^{-1}\notag\\
&\times\left[I-B_1(C_1B_1)^{-1}C_1\right]B_0\Sigma^{-1}\label{eq:F0}\\
F_1(z)&=zC_2(zI-\Gamma_1)^{-1}B_1(C_1B_1)^{-1}\label{eq:F1}\\
           &=C_2\Gamma_1(zI-\Gamma_1)^{-1}B_1(C_1B_1)^{-1}+C_2B_1(C_1B_1)^{-1}\label{eq:F1alt}
\end{align}
\end{subequations}
with $\Gamma_0$ given by \eqref{eq:Gamma0} and $\Gamma_1$ by
\begin{equation}
\label{eq:Gamma1}
\Gamma_1= \Gamma_0-B_1(C_1B_1)^{-1}C_1\Gamma_0.
\end{equation}
\end{theorem}
\begin{proof}
See Appendix~\ref{appendixthm4}
\end{proof}

\medskip

\begin{remark}\label{rem:Funique}
In view of \eqref{Fred},  $F(z)$ is uniquely determined by the decomposition \eqref{eq:zeta2uy} and the corresponding spectral density \eqref{Phidecomp}, so \eqref{F(z)} does not depend on the particular choice of spectral factor $W(z)$ used in constructing it.
\end{remark}

\medskip

\begin{corollary}\label{Fcor}
The transfer function $F(z)$ given by \eqref{F(z)} is (strictly) stable if and only if $\Gamma_1$ has all its eigenvalues in the (open) unit disc.
\end{corollary}

\medskip

\begin{proof}
Since $\Gamma_1(zI -\Gamma_1)^{-1} =(zI -\Gamma_1)^{-1} \Gamma_1$, it follows from \eqref{F0F1} that
\begin{subequations}\label{Falt}
\begin{equation}
F(z)=C_2(zI -\Gamma_1)^{-1} \hat{B}+\hat{D},
\end{equation}
where
\begin{align}
\label{}
 \hat{B}   &  =\begin{bmatrix}(I-B_1(C_1B_1)^{-1}C_1)B_0\Sigma^{-1}&\Gamma_1B_1(C_1B_1)^{-1}\end{bmatrix} \\
 \hat{D}   &  =\begin{bmatrix}0&C_2B_1(C_1B_1)^{-1}\end{bmatrix},
\end{align}
\end{subequations}
and consequently the corollary follows.
\end{proof}

\medskip

Note that, since $C_1\in\mathbb{R}^{(m-\rho)\times n}$ and $B_1\in\mathbb{R}^{n\times (m-\rho)}$, it is necessary that $m-\rho\leq n$ for $C_1B_1$ to be nonsingular.  Clearly, the McMillan degree of $F(z)$ is at most $n$. In special cases to be considered below, the McMillan degree  will depend on the rank of the matrix  $\Gamma_1$. To this end, we shall need the following lemma.
\medskip

\begin{lemma}\label{lem:Gamma1rank}
Suppose that $\rho<m$ and that $B_1(C_1B_1)^{-1}C_1$ has $n$ linearly independent eigenvectors. Then
\begin{displaymath}
\text{\rm rank}\,\Gamma_1 \leq n-(m-\rho).
\end{displaymath}
\end{lemma}

\medskip

\begin{proof}
By Lemma~\ref{similaritylem} in the appendix, the nonzero eigenvalues of $B_1(C_1B_1)^{-1}C_1$ are the same as those of $C_1B_1(C_1B_1)^{-1} = I_{m-\rho}$, so $B_1(C_1B_1)^{-1}C_1$ has $m-\rho$ nonzero eigenvalues all equal to 1. Then there is an $n\times n$ matrix $T$ such that
\begin{displaymath}
T^{-1}B_1(C_1B_1)^{-1}C_1T= \begin{bmatrix}I_{(m-\rho)}&  \\& 0_{n-(m-\rho)}\end{bmatrix},
\end{displaymath}
and therefore
\begin{displaymath}
T^{-1}(I-B_1(C_1B_1)^{-1}C_1)T=\begin{bmatrix}0_{(m-\rho)}&  \\&I_{n-(m-\rho)}\end{bmatrix}.
\end{displaymath}
Then, since $\Gamma_1=(I-B_1(C_1B_1)^{-1}C_1)\Gamma_0$, the statement of the lemma follows.
\end{proof}

\medskip

We summarize by formulating a procedure for  calculating the matrix function $F(z)$ in the model \eqref{FuHyv=0} from a minimal realization of any minimal stable spectral factor $W(z)$ of the spectral density $\Phi(z)$ of $\zeta(t)$.

\medskip

%{\bf Computional procedure:}
\begin{computation}\label{comp}
\phantom{xxxxxxxxxxxxxx}
\begin{enumerate}
  \item Perform the transformations \eqref{D2Sigma}.
  \item Rearrange the last $p-\rho$ components of $\zeta$ so that the square $(m-\rho)\times (m-\rho)$ matrix $C_1B_1$ is nonsingular. In general this can be done in several different ways.
  \item Determine $F(z)$ from \eqref{F(z)}, \eqref{F0F1} or \eqref{Falt}.
\end{enumerate}
\end{computation}

\subsection{The special case $D=0$}\label{subsec:D=0}

When $D=0$, we have $\rho=0$, and hence $C_0=B_0=0$. More precisely,
\begin{equation}
\label{CBpartitioningD=0}
C=\begin{bmatrix}C_1\\C_2\end{bmatrix}\qquad B=B_1.
\end{equation}
Then Step 1) in Procedure~\ref{comp} is not required,
and Step 3) is simplified. In fact,
$F_0(z)=0$ and hence $F(z)=F_1(z)$. Moreover,  $\Gamma_0=A$ and

\begin{displaymath}
\Gamma_1=\left[I-B(C_1B)^{-1}C_1\right]A.
\end{displaymath}
Consequently, \eqref{eq:F1} yields
\begin{subequations}\label{Ffor D=0}
\begin{equation}
\label{eq:FforD=0}
F(z)=zC_2(zI-\Gamma_1)^{-1}B(C_1B)^{-1}
\end{equation}
or alternatively from \eqref{eq:F1alt}
\begin{equation}
\label{eq:FaltforD=0}
F(z)=C_2\Gamma_1(zI-\Gamma_1)^{-1}B(C_1B)^{-1}+C_2B(C_1B)^{-1}.
\end{equation}
\end{subequations}

The realizations \eqref{Ffor D=0} are in general not minimal, as  under the conditions of Lemma~\ref{lem:Gamma1rank}, $\Gamma_1$ has at least one zero eigenvalue, and hence a pole in zero which will cancel the factor $z$ in \eqref{eq:F1}.  In fact, by the next theorem, all zero eigenvalues will be cancelled, and the McMillan degree of $F(z)$ will be reduced accordingly.

\medskip

\begin{theorem}\label{McMillanthm}
Suppose $B(C_1B)^{-1}C_1$ has $n$ linearly independent eigenvectors. Then $F(z)$ has McMillan degree at most $n-m$.
\end{theorem}

\medskip

\begin{proof}
The observability matrix of \eqref{eq:FaltforD=0} is
\begin{displaymath}
\begin{bmatrix}
C_2\Gamma_1\\
C_2(\Gamma_1)^2\\
\vdots\\
C_2(\Gamma_1)^{(p-m)}
\end{bmatrix}
=
\begin{bmatrix}
C_2\\
C_2\Gamma_1\\
\vdots\\
C_2(\Gamma_1)^{p-m-1}
\end{bmatrix}
\Gamma_1 ,
\end{displaymath}
which has at most the same rank as $\Gamma_1$. However, by Lemma~\ref{lem:Gamma1rank}, $\text{\rm rank}\,\Gamma_1<n-m$, so the realization \eqref{eq:FaltforD=0} is not observable and hence not minimal. In fact, the dimension of the unobservable subspace is at least $m$, so the dimension of $F(z)$ can be reduced from $n$ to $n-m$.
\end{proof}

\medskip

If $m = n$, $C_1$ and $B$ are both $m\times m$ matrices. Therefore, since rank$(C_1B)=m$,  they must both be invertible. Then $ \Gamma_1 = (I - B(C_1B)^{-1}C_1)\Gamma_0 = 0$, and consequently
    $$F(z) = C_2C_1^{-1}$$
is constant and hence strictly stable.

\subsection{The special case $D$ full rank}

When $D$ has full rank, $\rho=m$ and $C_1=B_1=0$, and therefore
\begin{displaymath}
C=\begin{bmatrix}C_0\\C_2\end{bmatrix}\qquad B=B_0.
\end{displaymath}
Then Step 2) in Procedure~\ref{comp} is not needed. Moreover,
$F_1(z)=0$ and hence $F(z)=F_0(z)$,
\begin{displaymath}
\Gamma_1=\Gamma_0=A-B\Sigma^{-1}C_0,
\end{displaymath}
and
\begin{equation}
\label{eq:FforDfullrank}
F(z)= C_2(zI-\Gamma_0)^{-1}B\Sigma^{-1}.
\end{equation}

\section{Causality and stability}\label{sec:stabilitycausality}

The ordering of the element of $\zeta$ in the decomposition \eqref{eq:zeta2uy} is in general not unique, and different choices may create feedback models with different stability and causality properties.

\subsection{Stability of $F(z)$}\label{ss:stability}

For the process $\zeta$ to be stationary, the feedback configuration in Figure~\ref{fbdtfig} needs to be internally stable \cite{DFT}. However, $F(z)$ does not need to be stable, as the feedback model can be stabilized by feedback.

The following simple counterexample answers Manfred Deistler's  question mentioned at the end of  Sect. \ref{sec:feedback},  in  the negative.

Let $\Phi$ be a spectral density  with
$$
  A = \begin{bmatrix}
        \frac{3}{2} & 2\\
        -1 & -\frac{3}{2}
        \end{bmatrix},
    B = \begin{bmatrix}
           2 \\ -3
        \end{bmatrix},
    C = \begin{bmatrix}
            -4 & -2\\
            -2 & 3
        \end{bmatrix},  D=\begin{bmatrix}
                                        0 \\0
                                      \end{bmatrix},
    $$
which has rank $m=1$. This is a process of the type studied in subsection \ref{subsec:D=0}, where $C$ is decomposed as \eqref{CBpartitioningD=0}.  There are two choices of ordering of the components of $\zeta$.

First choose $u=\zeta_1$ and $y=\zeta_2$. Then
$$
        C_1 = \begin{bmatrix}
            -4 & -2
        \end{bmatrix},\quad
        C_2 = \begin{bmatrix}
            -2 & 3
        \end{bmatrix},
    $$
and
    $$
   \Gamma_1 = (I - B(C_1B)^{-1}C_1)A
            = \begin{bmatrix}
            -\frac{5}{2} & -3 \\
            5 & 6 \end{bmatrix},
    $$
which has rank 1 with eigenvalue $0$ and $\frac{7}{2}$. Moreover,
  \begin{equation}
\label{unstableF1}
F(z) = \frac{26z-27}{2(2z-7)},
\end{equation}
which is unstable.

Next, choose $u=\zeta_2$ and $y=\zeta_1$. Then
$$
        C_1 = \begin{bmatrix}
            -2 & 3
        \end{bmatrix},\quad
        C_2 = \begin{bmatrix}
            -4 & -2
        \end{bmatrix},
    $$
and
    $$
   \Gamma_1 = (I - B(C_1B)^{-1}C_1)A
            = \begin{bmatrix}
            \frac{15}{26} & \frac{9}{13} \\
            \frac{5}{13} & \frac{6}{13} \end{bmatrix},
    $$
which has rank 1 with eigenvalue $0$, $\frac{27}{26}$. This yields
 \begin{equation}
\label{unstableF2}
   F(z) = \frac{2(2z-7)}{26z-27},
\end{equation}
which is again unstable.

Consequently, there is no selection of the order of the components in $\zeta$ admitting a stable $F(z)$.  Indeed, $F(z)$ depends only on $\Phi(z)$ and not on the particular choice of spectral factor \eqref{W} (Remark~\ref{rem:Funique}), so we only have the two $F(z)$ obtained above.

\subsection{Granger causality}

If we want to predict the future of $y$ given the past of $y$, would we get a better estimate if we also know the past of $u$? If so, we have {\em Granger causality from $u$ to $y$} \cite{Granger63,Granger69,BS17,PD12,ADD19}. Let us consider the negative situation that there is no such advantage. In mathematical terms, we have non-causality if and only if
\begin{equation}
\label{noncausality}
\mathbb{E}^{\Hb_t^-(y)\vee\Hb_t^-(u)}\lambda =\mathbb{E}^{\Hb_t^-(y)}\lambda \quad \text{for all $\lambda\in\Hb_t^+(y)$}
\end{equation}
\cite[Definition 1]{Granger69}, where $\mathbb{E}^{\mathbf{A}}\lambda$ denotes the orthogonal projection of $\lambda$ onto the subspace $\mathbf{A}$ and $\vee$ is vector sum, i.e., $\mathbf{A}\vee\mathbf{B}$ is the closure in the Hilbert space of stochastic variables of the sum of the subspaces $\mathbf{A}$ and $\mathbf{B}$; see, e.g., \cite{LPbook}. With $A\ominus B$ the orthogonal complement of $B\subset A$ in $A$, \eqref{noncausality} can also be written
\begin{displaymath}
\mathbb{E}^{\Hb_t^-(y)}\lambda +\mathbb{E}^{[\Hb_t^-(y)\vee\Hb_t^-(u)]\ominus\Hb_t^-(y)}\lambda =\mathbb{E}^{\Hb_t^-(y)}\lambda
\end{displaymath}
for all $\lambda\in\Hb_t^+(y)$, which is equivalent to
\begin{displaymath}
\left[\Hb_t^-(y)\vee\Hb_t^-(u)\right]\ominus\Hb_t^-(y)\perp \Hb_t^+(y),
\end{displaymath}
where $\mathbf{A}\perp\mathbf{B}$ means that the subspaces $\mathbf{A}$ and $\mathbf{B}$ are orthogonal. Then, using the equivalence between properties (i) and (v) in \cite[Proposition 2.4.2]{LPbook}, we see that this in turn is equivalent to the following geometric condition for lack of Granger causality
\begin{equation}
\label{condorth}
\Hb_t^-(u)\perp \Hb_t^+(y)\mid \Hb_t^-(y),
\end{equation}
i.e., $\Hb_t^-(u)$ and $\Hb_t^+(y)$ are conditionally orthogonal given $\Hb_t^-(y)$. Hence, if the past of $y$ is known, the future of $y$ is uncorrelated to the past of $u$, and therefore $$\mathbb{E}\{y(t)\mid \Hb_t^-(u)\}=0,$$  so, in view of \eqref{u2y}, lack of Granger causality is equivalent to $F(z)\equiv 0$. Conversely, we have Granger causality from $u$ to $y$ if and only if $F(z)$ is nonzero.

An analogous argument applied to \eqref{y2u} yields the the geometric condition
\begin{equation}
\label{feedbackfree}
\Hb^-(y)\perp \Hb^+(u)\mid \Hb^-(u),
\end{equation}
which is equivalent to $H(z)\equiv 0$. Then there is no feedback from $y$ to $u$  \cite[p. 677]{LPbook}. Consequently, as stressed in \cite{Caines}, Granger causality and feedback are dual concepts.
In the setting of Section~\ref{sec:Fconstruction} we must have $H(z)$ nonzero if $F(z)$ is not strictly stable, because it is needed for stabilization of the feedback loop. Conversely, if $H(z)$ is zero, $F(z)$ must be strictly stable.

\medskip

\begin{theorem}
Consider the feedback model \eqref{FuHy}, and in particular, \eqref{FuHyv=0}. Then there is causality from $u$ to $y$ in  the sense of Granger if and only if $F(z)$ is nonzero, and there is no feedback from $y$ to $u$ if and only if $H(z)$ is identically zero. In this case $F(z)$ is (strictly) stable.
\end{theorem}

\section{Singular Dynamic Network Models}\label{sec:networks}

From \eqref{rv2uy} and \eqref{eq:w2rv} we have
\begin{subequations}\label{networkmodel}
\begin{equation}
\label{ }
\zeta(t)=M(z)\zeta(t)+N(z)w(t),
\end{equation}
where
\begin{equation}
\label{ }
M(z)=\begin{bmatrix}0 &H(z)\\F(z) & 0\end{bmatrix}, \quad N(z)=\begin{bmatrix}K(z)&0\\0&G(z)\end{bmatrix}.
\end{equation}
\end{subequations}
We may choose the $p\times m$ matrix $N(z)$ to be stable and have a left stable inverse. Then with the diagonal elements of $M(z)$ all identically zero, \eqref{networkmodel} corresponds to a dynamic network model with a noise term $N(z)w(t)$ and no exogenous input \cite{WVD18}.
For a more detailed description we could choose
\begin{align*}
  M(z) & = \begin{bmatrix}
             0 & M_{12} & \cdots & M_{1p} \\
             M_{21} & 0 & \cdots & M_{2p} \\
             \vdots & \vdots & \ddots & \vdots \\
             M_{p1} & M_{p2} & \cdots & 0
           \end{bmatrix},
\end{align*}
but we shall stick with to  the formulation with blocks of  zeros.
Such  models describe the dynamical dependencies between components of a multivariate stationary stochastic process
in terms of a network whose links are dynamical relations. They play an important role in understanding the underlying
mechanisms of complex systems in econometrics, biology and engineering \cite{nweco10,nwbio09, nweng91}.

To connect the network model \eqref{networkmodel} to the the feedback representation of Section~\ref{sec:feedback} we observe that
\begin{displaymath}
\zeta(t)=(I-M(z))^{-1}N(z)w(t)
\end{displaymath}
and that
\begin{displaymath}
N(z)w(t)=\begin{bmatrix} r(t)\\v(t)\end{bmatrix},
\end{displaymath}
and consequently
\begin{equation}
\label{ }
(I-M(z))^{-1}=T(z),
\end{equation}
which, by Lemma~\ref{lem1},  is strictly stable and has a representation \eqref{TPQ}.

In the following, we shall discuss the relation between the special feedback structure \eqref{FuHyv=0} and singular dynamic networks, and show that a dynamic network can be simplified by using \eqref{FuHyv=0}. As $v=0$ we take
\begin{displaymath}
N(z)=\begin{bmatrix}K(z)\\0\end{bmatrix}
\end{displaymath}
and $N(z)w(t)=r$. This satisfies a necessary condition of the identifiability of a singular dynamic network \cite{WVD18}. Different from the situation in \cite{WVD18}, by using this simplified model, the identifiability and identification of a singular dynamic network can be directly obtained from our research on identification to low rank processes \cite{CPLauto21}.

Recovering or reconstructing the topology of a dynamic network by identifying the matrix $M(z)$ is important when the prior information about the topology is scarce, or some nodes are not measurable. Considering the sparsity propoerty and interconnections in a large-scale dynamic network, singular dynamic network models are increasingly popular these days, see e.g., \cite{BCV18, GLsampling,WVD18, BGHP17, Ferrante-18letters}.

We will give a simple example of modeling the topology and the transfer matrix $M(z)$ for a singular dynamic network.
Suppose $\zeta=[\zeta_1, \zeta_2, \zeta_3,\zeta_4]'$ has $4$ scalar nodes of rank $2$, with $u(t) :=[\zeta_1(t), \zeta_2(t)]'$ full rank. Then with $y(t) :=[\zeta_3(t), \zeta_4(t)]'$, \eqref{FuHyv=0} becomes
\begin{align*}
  \begin{bmatrix}
    \zeta_3(t)  \\
    \zeta_4(t)
  \end{bmatrix} &=F(z) \begin{bmatrix}
    \zeta_1(t)  \\
    \zeta_2(t)
  \end{bmatrix},\\
  \begin{bmatrix}
    \zeta_1(t)  \\
    \zeta_2(t)
  \end{bmatrix}&=H(z) \begin{bmatrix}
    \zeta_3(t)  \\
    \zeta_4(t)
  \end{bmatrix} + K(z)w(t),
\end{align*}
where $K(z)$ is full rank $2\times 2$, and $w(t)$ is a white noise of dimension $2$. Suppose that, after calculation, $F(z)$ is diagonal and $H(z)$ is upper triangular.
Then, a simplified inner topology can be constructed for $\zeta(t)$, compared with a general one. As shown in Fig.~\ref{fig:DNM}, by introducing the special feedback structure, the inner topology of this dynamic network can be simplified from having possibly $12$ edges to only $5$ edges.

\begin{figure}[!ht]
  \centering
  \includegraphics[scale=0.44]{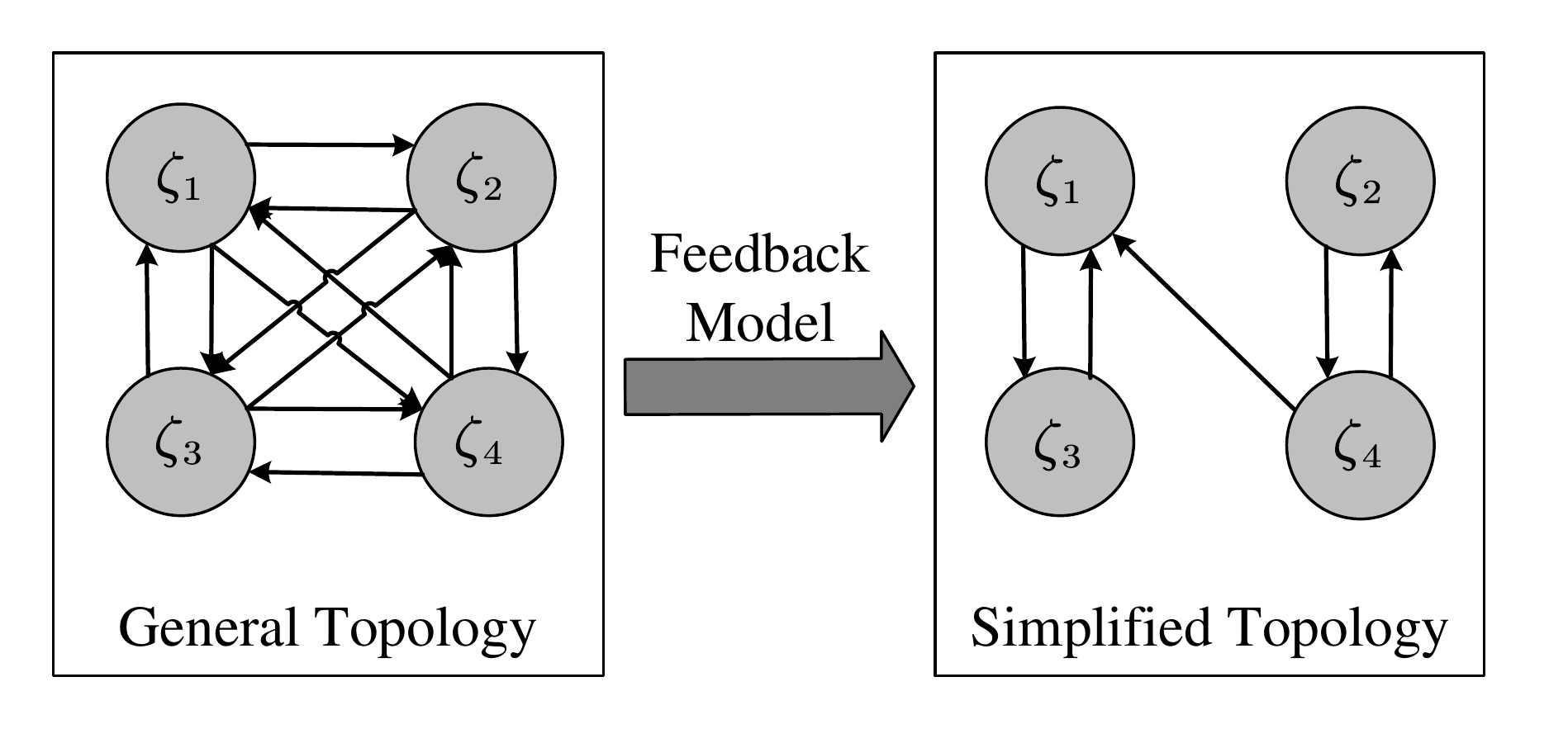}
  \caption{The inner topology of a dynamic network in the simple example.}\label{fig:DNM}
\end{figure}

\section{Connections to dynamic factor models}\label{sec:factormodels}

Suppose that we want to model a large dimensional stationary process  $\{\eta(t), t\in\mathbb{Z}\}$, assumed zero-mean and of full rank by a {\em Dynamic Factor Analysis } model. This amounts to decomposing its  spectral density, say $\Psi(z)$,   into a sum of a low-rank spectral density $\Phi(z)$ and a diagonal full-rank spectral density $\Delta(z)$ which in principle should be diagonal, although  this condition has been somewhat relaxed in the literature  \cite{Forni-L-2001,Forni-H-L-R-2000,Stock-W-2010}. This corresponds to the decomposition
\begin{displaymath}
\eta(t) =\zeta(t) +\omega(t),
\end{displaymath}
where $\{\omega(t), t\in\mathbb{Z}\}$ is a full rank noise process with uncorrelated components and $\{\zeta(t), t\in\mathbb{Z}\}$, called  a {\em latent process}, has density $\Phi(z)$ having (hopefully very) low rank $m<n$.  By possibly rearranging the components of $\zeta(t)$, this latent process can  be decomposed (in several ways) into two components as in \eqref{eq:zeta2uy}, i.e.,
\begin{equation*}
%\label{eq:zeta2uy}
\zeta(t) =\begin{bmatrix}u(t)\\y(t)\end{bmatrix},
\end{equation*}
where $u(t)$ is chosen of full rank $m$. Then $\zeta(t)$ can be modeled    by a feedback system of a special structure  as    \eqref{FuHyv=0}. We shall study the latent process in this context.

 Then
\begin{equation}
\label{zetaWu}
\zeta = \begin{bmatrix} I\\F(z)\end{bmatrix}u,
\end{equation}
so that $u$ plays the role of a {\em factor process} of minimal dimension and \eqref{Phidecomp} becomes
\begin{equation}
\label{Phifactorization}
\Phi(z)=\begin{bmatrix} I\\F(z)\end{bmatrix}\Phi_u(z)\begin{bmatrix} I\\F(z^{-1})\end{bmatrix}'.
\end{equation}
where $\Phi_u(z)$ is full rank. Note that the factor process may not be unique when actually doing estimation, and, even once the decomposition \eqref{eq:zeta2uy} is fixed,  it is {\em apriori} not clear how it may be constructed from the data.

There  has been a widespread interest in estimating $u$ from the observable data \cite{Doz-G-R-2011,Ferrante-20cdc}.  Now the feedback representation \eqref{FuHyv=0} provides a partial answer to this question as it shows that:

\medskip
\begin{corollary}
Every minimal factor process $u$ can be constructed by a noisy feedback
$$
u(t)= H(z) y(t) + r(t)
$$
on the ``dependent" (or residual) variables $y(t)$.
\end{corollary}
\medskip
Note again that there is non-uniqueness in this representation. In particular there are infinitely many pairs $(F,H)$ which yield the same transfer function $r\rightarrow y$ of the feedback system
 \eqref{FuHyv=0} and hence the same spectral density $\Phi(z)$.

In view of \eqref{Phifactorization} the spectral factor \eqref{W} of $\zeta$ can be written
\begin{equation}
\label{Wdecomp}
W(z)=\begin{bmatrix}W_u(z)\\W_{yu}(z)\end{bmatrix}=\begin{bmatrix} I\\F(z)\end{bmatrix}W_u(z),
\end{equation}
with
\begin{displaymath}
\Phi_u(z)=W_u(z)W_u(z^{-1})',
\end{displaymath}
where $W_u(z)$ is a stable spectral factor, and
\begin{equation}
\label{ }
F(z)=W_{yu}(z)W_u(z)^{-1}.
\end{equation}
As mentioned above, $u$ plays the role of a minimal dynamic factor \cite{EJC,Deistler19}. Moreover,
\begin{equation}
\label{ }
\begin{bmatrix}-F(z)&I\end{bmatrix}W(z)=0,
\end{equation}
so $\begin{bmatrix}-F&I\end{bmatrix}$ is the rational matrix function whose rows form a basis for the left kernel of $W$; cf. \cite[Section 5]{Deistler19}.
This configuration is illustrated in Figure \ref{fig:WuF}, where $w$ is the generating white noise in the realization \eqref{eq:system}. More precisely, the transfer function $W_y(z)$ from $w$ to $y$  is a cascade of two  transfer functions which we can compute.
\begin{figure}[h]%[thpb]
      \centering
      \includegraphics[scale=0.5]{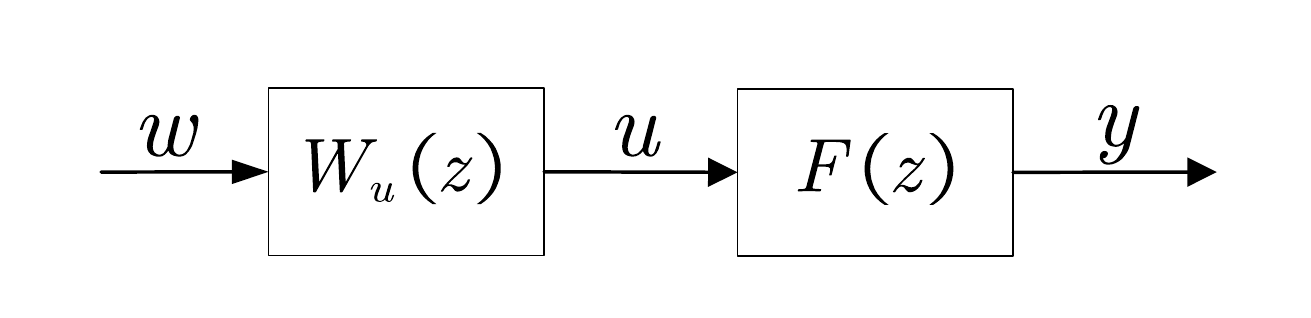}
      \caption{Dynamical relation from $w$ to $u$ to $y$}
      \label{fig:WuF}
\end{figure}

Introducing the decompositions
\begin{displaymath}
C=\begin{bmatrix}C_u\\C_{yu}\end{bmatrix}, \quad D=\begin{bmatrix}D_u\\D_{yu}\end{bmatrix}
\end{displaymath}
in the format of \eqref{Wdecomp}, the latent process $\zeta$ has the representation \eqref{zetaWu} with the square $m\times m$ spectral factor $W_u(z)$ having the realization
\begin{equation}
\label{Wureal}
W_u(z)=C_u(zI-A)^{-1}B + D_u,
\end{equation}
where, in the notation of Section~\ref{sec:Fconstruction},
\begin{equation}
\label{ }
C_u=\begin{bmatrix}C_0\\C_1\end{bmatrix}, \quad D_u=\begin{bmatrix}\Sigma&0\\0&0\end{bmatrix}.
\end{equation}
In particular,
\begin{equation}
\label{WuD=0}
W_u(z) =C_1(zI-A)^{-1}B
\end{equation}
in the special case that $D=0$, and
\begin{equation}
\label{ }
W_u(z) =C_0(zI-A)^{-1}B  +\Sigma
\end{equation}
when $D$ is full rank.

Let us illustrate this with the simple example in subsection~\ref{ss:stability}, taking the first choice of ordering of the components in $\zeta$, namely $u=\zeta_1$ and $y=\zeta_2$.  Then $ C_1 = \begin{bmatrix}-4 & -2\end{bmatrix}$, so \eqref{WuD=0} yields
\begin{displaymath}
W_u(z)=\frac{28-8z}{4z^2-1},
\end{displaymath}
which together with \eqref{unstableF2} yields the transfer functions in Figure~\ref{fig:WuF}.   If instead we choose $u=\zeta_2$ and $y=\zeta_1$,
\begin{displaymath}
W_u(z)=\frac{54-52z}{4z^2-1},
\end{displaymath}
and $F(z)$ given by \eqref{unstableF2}.

\section{Examples}\label{sec:examples}

We begin by giving examples illustrating the results in Section~\ref{sec:Fconstruction}.
We shall consider three different situations, namely that $D=0$ and $D$ is full rank,  as well as the mixed case when $0< \rank D < m$.
Finally, we shall give an example for how to determine $H(z)$ when $Q(z)$ is a matrix.

\subsection{Example 1: $D=0$}
Let $\Phi(z)$ be a spectral density with $D=0$ and
$$
    A = \begin{bmatrix}
        1 & 0 & -\frac{3}{2}\\
        \frac{7}{10} & \frac{1}{5} & -\frac{7}{5}\\
        \frac{1}{2} & 0 &-1
        \end{bmatrix},\;
    B = \begin{bmatrix}
            1 \\ 1 \\ -3
        \end{bmatrix},\;
    C = \begin{bmatrix}
            3 & -3 & -3\\
            2& 0& -1\\
            0& 1& 0\\
            0 & 0 & 1
        \end{bmatrix},
$$
which has rank $m=1$. First take $C_1$ to be the first row of $C$, i.e., $u=\zeta_1$,  $y=(\zeta_2,\zeta_3,\zeta_4)'$, and
$$
        C_1 = \begin{bmatrix}
            3 & -3 & -3
        \end{bmatrix},\quad
        C_2 = \begin{bmatrix}
            2& 0& -1\\
            0& 1& 0\\
            0 & 0 & 1
        \end{bmatrix}.
$$
Then, $C_1B=9$, $B(C_1B)^{-1}C_1$ has $n=3$ independent eigenvectors, and
    $$
            \Gamma_1 = (I - B(C_1B)^{-1}C_1)A
            = \begin{bmatrix}
            \frac{16}{15} & \frac{1}{15} & -\frac{9}{5}\\
            \frac{23}{30} & \frac{4}{15} & -\frac{17}{10}\\
            \frac{3}{10} & -\frac{1}{5} & -\frac{1}{10} \end{bmatrix},
    $$
which has rank two with eigenvalues $\frac{9}{10}$, $\frac{1}{3}$ and $0$. However, by Theorem~\ref{McMillanthm}, the pole at zero will cancel, and we obtain
$$
    F(z)=  \frac{1}{3(10z-9)(3z-1)}\begin{bmatrix}
                    5(2z+3)(5z-1) \\ 10z^2 + 49z - 13 \\ (7 - 6z)(5z-1)
                    \end{bmatrix}
$$
which is strictly stable of degree two rather than three.

Since the McMillan degree of $F(z)$ is two, it has a minimal realization of dimension two. One such realization is given by
\begin{displaymath}
F(z)=\tilde{C}(zI-\Gamma)^{-1}\tilde{B} +\tilde{D},
\end{displaymath}
where
$$
    \Gamma = \begin{bmatrix} \frac{1}{3} & 0\\ 0 &  \frac{9}{10}\end{bmatrix},
        \tilde{B} = \begin{bmatrix} -\frac{10}{459}\\\frac{28}{17} \end{bmatrix},
        \tilde{C} = \begin{bmatrix} 11 & 1\\ 4 & \frac{7}{15}\\ 3 & \frac{1}{15} \end{bmatrix},
    \tilde{D} = \begin{bmatrix} \frac{5}{9} \\ \frac{1}{9} \\ -\frac{1}{3}\end{bmatrix}.
$$

Next, take $C_1$ to be the second row of $C$, i.e., $u=\zeta_2$,  $y=(\zeta_1,\zeta_3,\zeta_4)'$, and
$$
        C_1 = \begin{bmatrix}
            2& 0& -1
        \end{bmatrix},\quad
        C_2 = \begin{bmatrix}
            3& -3& -3\\
            0& 1& 0\\
            0 & 0 & 1
        \end{bmatrix}.
    $$
Then, $C_1B=5$, and
$$
        \Gamma_1 = (I - B(C_1B)^{-1}C_1)A
            = \begin{bmatrix}
            \frac{7}{10} & 0 & -\frac{11}{10}\\
            \frac{2}{5} & \frac{1}{5} & -1\\
            \frac{7}{5} & 0 & -\frac{11}{5} \end{bmatrix},
$$
which has rank two with eigenvalue $-\frac{3}{2}$, $\frac{1}{5}$ and $0$. In this case, $B(C_1B)^{-1}C_1$ has only two independent eigenvalues, so we cannot apply Theorem~\ref{McMillanthm}. However, due to \eqref{eq:FforD=0}, the zero pole will nevertheless cancel, and we obtain
$$
    F(z)=  \frac{1}{5(2z+3)(5z-1)}\begin{bmatrix}
                    3(10z-9)(3z-1) \\ 10z^2 + 49z - 13 \\ (7 - 6z)(5z-1)
                    \end{bmatrix}
$$
which is unstable of degree two rather than three. The system
  \begin{displaymath}
F(z)=\tilde{C}(zI-\Gamma)^{-1}\tilde{B} +\tilde{D},
\end{displaymath}
with
  $$\Gamma = \begin{bmatrix} -\frac{3}{2} & 0\\ 0 & \frac{1}{5} \end{bmatrix},\;
        \tilde{B} = \begin{bmatrix} \frac{8}{5}\\ \frac{14}{425} \end{bmatrix},\;
      \tilde{C} = \begin{bmatrix} -\frac{99}{34} & 3\\ \frac{8}{17} & -1\\ 1 & 0 \end{bmatrix},\;
      \tilde{D} = \begin{bmatrix} \frac{9}{5} \\ \frac{1}{5} \\ -\frac{3}{5} \end{bmatrix}$$
is a minimal realization of $F(z)$.

\subsection{Example 2: $D$ full rank}
\subsubsection{Example 2.1}
Given the spectral density $\Phi(z)$, let
$$
\bar{W}(z)=\bar{C}(zI-A)^{-1}\bar{B}+\bar{D}
$$
be a spectral factor,
where
$$
     A = \begin{bmatrix}
        -\frac{1}{5} & 0 & 0\\
        \frac{1}{2} & \frac{9}{20} & \frac{1}{20}\\
        -\frac{13}{10} & \frac{1}{5} & \frac{3}{10}
        \end{bmatrix},\quad
   \bar{B} = \begin{bmatrix}
            \frac{1}{\sqrt{2}} & -\frac{1}{\sqrt{2}} \\
             \frac{1}{\sqrt{2}} & -\frac{3}{\sqrt{2}} \\
             \frac{1}{\sqrt{2}} & \frac{1}{\sqrt{2}}
        \end{bmatrix},\;
$$
$$
    \bar{C} = \begin{bmatrix}
       -3 & -1 & -1 \\
       -2\sqrt{2} & -\frac{5}{\sqrt{2}} & \frac{1}{\sqrt{2}} \\
       0 & \frac{1}{\sqrt{2}} & -\frac{7}{\sqrt{2}}
     \end{bmatrix}, \quad
     \bar{D}=\begin{bmatrix}
                     \frac{1}{\sqrt{2}} & \frac{1}{\sqrt{2}} \\
                     1 & -1 \\
                     -1 & 1
                   \end{bmatrix},
$$
with $\rho = m=2$ and  $\bar{D}$ full column rank.

Performing the transformations \eqref{D2Sigma} on this system with
$$
U=\begin{bmatrix}
    0 & \frac{1}{\sqrt{2}} & -\frac{1}{\sqrt{2}} \\
    1 & 0 & 0 \\
    0 & \frac{1}{\sqrt{2}} & \frac{1}{\sqrt{2}}
  \end{bmatrix},\quad
V'=
 \begin{bmatrix}
     \frac{1}{\sqrt{2}} & \frac{1}{\sqrt{2}} \\
     -\frac{1}{\sqrt{2}} & \frac{1}{\sqrt{2}}
   \end{bmatrix},
$$
 we have a new system with
$$
B=\begin{bmatrix}
    1 &  0 \\
    2 & -1 \\
    0 &  1
  \end{bmatrix},\,
            C=\begin{bmatrix}
                     -2 & -3 & 4 \\
                     -3 & -1 & -1 \\
                     -2 & -2 & -3
                   \end{bmatrix}, \,
                   D=\begin{bmatrix}
    2 &  0 \\
    0 & 1 \\
    0 &  0
  \end{bmatrix},
$$
and hence
$$
 \Sigma = \begin{bmatrix}
            2 & 0 \\
            0 & 1
          \end{bmatrix}.
$$

Then, we have $B=B_0$,
$$
C_0= \begin{bmatrix}
                     -2 & -3 & 4 \\
                     -3 & -1 & -1
                   \end{bmatrix},\quad
 C_2=\begin{bmatrix}
                     -2 & -2 & -3
                   \end{bmatrix},
 $$
and
$$
 \Gamma_0=A-B\Sigma^{-1}C_0 =
     \begin{bmatrix}
     \frac{4}{5} & \frac{3}{2} & -2 \\
     -\frac{1}{2} & \frac{49}{20} & -\frac{99}{20} \\
     \frac{17}{10} & \frac{6}{5} & \frac{13}{10}
     \end{bmatrix},
$$
which has full rank $3$. Hence \eqref{eq:FforDfullrank} is a minimal realization of $F(z)$, and
\begin{displaymath}
F(z)=\frac{-5}{\chi(z)} \begin{bmatrix}
    240z^2 + 56z - 67, & 4(20z^2 - 521z - 161)
    \end{bmatrix},
\end{displaymath}
where
\begin{displaymath}
\chi(z)=400z^3 - 1820z^2 + 6510z - 2073.
\end{displaymath}
By calculating the zeros of $\chi(z)$, we see thet $F(z)$ is unstable.

However, the following example shows that the McMillan degree of $F(z)$ may be strictly less than $n$ even when $D$ is full rank.

\subsubsection{Example 2.2}
Let $\Phi(z)$ be a spectral density with
$$
A=\begin{bmatrix}
    -\frac{1}{2} & 0 & 0 \\
    0 & \frac{7}{10} & 0 \\
    0 & 0 & \frac{1}{10}
  \end{bmatrix}, \quad B=\begin{bmatrix}
                     3 & -\frac{3}{2} \\
                     3 & -4 \\
                     1 & 1
                   \end{bmatrix},
$$
$$
C=\begin{bmatrix}
                                     1 & 1 & 0 \\
                                     1 & 1 & 0 \\
                                     0 & 0 & 1
                                   \end{bmatrix},\quad
D=\begin{bmatrix}
    1 & 0 \\
    0 & -1 \\
    0 & 0
  \end{bmatrix}.
$$
Here $D$ is already a tall diagonal matrix so Step 1) in Procedure~\ref{comp} is not needed, and thus
$$
 \Sigma=\begin{bmatrix}
         1 & 0 \\
        0 & -1
        \end{bmatrix}.
$$
Then
$$
    C_0=\begin{bmatrix}
          1 & 1 & 0 \\
          1 & 1 & 0
        \end{bmatrix}
$$
and
$$
\Gamma_0=A-B\Sigma^{-1}C_0 = \begin{bmatrix}
     -5 & -\frac{9}{2} & 0 \\
     -7 & -\frac{63}{10} & 0 \\
     0 & 0 & \frac{1}{10}
     \end{bmatrix},
$$
which has rank two. Hence $F(z)$ has a realization \eqref{eq:FforDfullrank} with  McMillan degree strictly less than $n=3$. In fact,
the observability and reachability matrices of this realization are
$$
\mathcal{O}=\begin{bmatrix}
              C_2 \\ C_2\Gamma_0 \\ C_2\Gamma_0^2
            \end{bmatrix}=\begin{bmatrix}
                            0 & 0 & 1 \\
                            0 & 0 & 1/10 \\
                            0 & 0 & 1/100
                          \end{bmatrix},
$$
\begin{equation}\nonumber
\begin{split}
\mathcal{R} &=\begin{bmatrix}
              B\Sigma^{-1} & \Gamma_0B\Sigma^{-1} & \Gamma_0^2B\Sigma^{-1}
            \end{bmatrix} \\
            &=\begin{bmatrix}
                            3 & \frac{3}{2} & -\frac{57}{2} & -\frac{51}{2} &  \frac{6441}{20} &  \frac{5763}{20} \\
                            3 & 4 & -\frac{399}{10}  & -\frac{357}{10} & \frac{10370}{23} & \frac{15733}{39}  \\
                            1 & -1 &   \frac{1}{10}  & -\frac{1}{10} & \frac{1}{100} & -\frac{1}{100}
        \end{bmatrix},
\end{split}
\end{equation}
so $\rank(\mathcal{O}\mathcal{R})=1$ and $F(z)$ has a minimal realization of dimension $1$, namely
\begin{displaymath}
F(z)=C_2(zI-\Gamma_0)^{-1}B\Sigma^{-1}=\frac{1}{z-1/10}\begin{bmatrix}1&-1\end{bmatrix}
\end{displaymath}
which is stable.

\subsection{Example 3: mixed case}
Let $\bar{\Phi}(z)$ be a spectral density with spectral factor
$$
\bar{W}(z)=\bar{C}(zI-A)^{-1}\bar{B}+\bar{D} ,
$$
where
$$
     A = \begin{bmatrix}
        -\frac{1}{2} & \frac{4}{5} & \frac{26}{5}\\
        0 & \frac{7}{5} & \frac{18}{5}\\
        0 & -\frac{3}{10} & -\frac{7}{10}
        \end{bmatrix},\quad
    \bar{B} = \begin{bmatrix}
            \frac{3}{\sqrt{2}} & 2 & -\frac{3}{\sqrt{2}} \\
             -2{\sqrt{2}} & 1& 0\\
             \frac{5}{\sqrt{2}} & 0 & -\frac{1}{\sqrt{2}}
        \end{bmatrix},\;
$$
$$
    \bar{C} = \begin{bmatrix}
       -1 & 4 & -1 \\
       1 & 3 & -2 \\
       3 & 0 & 0\\
       -4 & -1 & 0
     \end{bmatrix}, \quad
    \bar{D}=\begin{bmatrix}
     0 & -1 & 0 \\
                     -\frac{3}{\sqrt{2}} & 0& -\frac{3}{\sqrt{2}} \\
                     0 & 0 & 0 \\
                     0 & 0 & 0
                   \end{bmatrix},
$$
with $p=4$, $m= n=3$, $\rank(\bar{D})=\rho=2$.

Performing the transformations \eqref{D2Sigma} with
\begin{displaymath}
U=\begin{bmatrix}
    1 & 0 & 0 & 0 \\
    0 & 1 & 0 & 0 \\
    0 & 0 & 1 & 0 \\
    0 & 0 & 0 & 1
  \end{bmatrix}, ~ V'=\begin{bmatrix}
    0 & \frac{1}{\sqrt{2}} & -\frac{1}{\sqrt{2}} \\
    1 & 0 & 0 \\
    0 & \frac{1}{\sqrt{2}} & \frac{1}{\sqrt{2}}
  \end{bmatrix},
\end{displaymath}
we have the new system with
$$
B=\begin{bmatrix}
    2 &  0 & -3 \\
    1 & -2 & 2\\
    0 &  2 & -3
  \end{bmatrix}, \, C=\begin{bmatrix}
       -1 & 4 & -1 \\
       1 & 3 & -2 \\
       3 & 0 & 0\\
       -4 & -1 & 0
                   \end{bmatrix},
$$
 and
 $$
 D=\begin{bmatrix}
    -1 &  0 &  0\\
    0 & -3 &  0\\
    0 &  0 &  0 \\
     0 &  0 &  0
  \end{bmatrix}.
$$
Consequently
$$
 \Sigma=  \begin{bmatrix}
            -1 & 0 \\
            0 & -3
          \end{bmatrix}
$$
and
$$
  B_0= \begin{bmatrix}
    2 & 0 \\
    1 & -2 \\
    0 & 2
    \end{bmatrix},~ B_1=\begin{bmatrix}
                         -3 \\
                         2 \\
                         -3
                       \end{bmatrix},~  C_0=\begin{bmatrix}
       -1 & 4 & -1 \\
       1 & 3 & -2
       \end{bmatrix}.
$$

Choose $C_1$ to be the first row among the last $p-\rho=2$ rows of $C$, i.e., the third row of $C$,
$$
  C_1= \begin{bmatrix}
         3 & 0 & 0
       \end{bmatrix}
$$
with $C_1B_1=-9$ full rank.
Then by \eqref{eq:Gamma0} and \eqref{eq:Gamma1},
$$
\Gamma_0=A-B_0\Sigma^{-1}C_0=
    \begin{bmatrix}
    -\frac{5}{2} & \frac{44}{5} & \frac{16}{5} \\
    -\frac{5}{3} & \frac{17}{5} & \frac{59}{15} \\
    \frac{2}{3} & \frac{17}{10} & -\frac{61}{30}
    \end{bmatrix},
$$
$$
\Gamma_1= \Gamma_0-B_1(C_1B_1)^{-1}C_1\Gamma_0=
    \begin{bmatrix}
    0 & 0 & 0 \\
    -\frac{10}{3} & \frac{139}{15} & \frac{91}{15} \\
    \frac{19}{6} & -\frac{71}{10} & -\frac{157}{30}
    \end{bmatrix},
$$
with rank $3$ and $2$, respectively. Therefore we have $F(z)= \begin{bmatrix} F_0(z) & F_1(z)\end{bmatrix}$ with
\begin{align*}
 F_0 &= \frac{1}{90z^2-363z-488}\begin{bmatrix}
                                7(30z+1) & 10(5-6z)
                              \end{bmatrix},\\
 F_1 &=\frac{- 300z^2 + 1520z + 1793}{3(90z^2-363z-488)}.
\end{align*}
%{\color{red} Is it unstable?}
 Hence $F(z)$ is unstable with two poles, namely $(121-\sqrt{34161})/60$ and $(121+\sqrt{34161})/60$.

\subsection{Example 4: determining $H(z)$ in the matrix case} \label{Example4}
In this example, we first use coprime factorization (see, e.g., \cite{Francis}) to obtain the interpolation conditions as in \cite[p. 2174]{BLN}, and then solve the Nevanlinna-Pick interpolation problem by the approach in \cite{CLtac}.

Suppose that Theorem~\ref{thm:F(z)} has given us
$$
     F(z)=\begin{bmatrix}
                        \frac{z+3}{z+2} & 0 \\
                        0 & \frac{z-4}{z-2}
                      \end{bmatrix}.
$$
Then, by the bilinear Tustin transformation, we have the corresponding function
\begin{equation}
\label{G(s)}
    G(s)=F\left(\frac{1+s}{1-s}\right)=\begin{bmatrix}
                                            \frac{2s-4}{s-3} & 0 \\
                                            0 & \frac{5s-3}{3s-1}
                                          \end{bmatrix}
\end{equation}
 in the $s$-domain.
Suppose we want to find a function $K(s)$ so that the sensitive function
\begin{equation}
    S(s)=(I-GK)^{-1}
\end{equation}
is stable. Then the discrete-time sensitivity function $Q(z)$ can be obtained by performing the Tustin transformation. If one wants to use the input sensitive function $P(z)$ instead, $S(s)=(I-KG)^{-1}$ should be used instead.

By \cite[Lemma1, p. 23]{Francis},
\begin{equation}
\label{G}
G(s)= N(s)M(s)^{-1}=\tilde{M}(s)^{-1}\tilde{N}(s)
\end{equation}
with
\begin{subequations}
\begin{equation}
\label{coprime1}
\tilde{X}M-\tilde{Y}N=I
\end{equation}
\begin{equation}
\label{coprime2}
\tilde{M}X-\tilde{N}Y=I,
\end{equation}
\end{subequations}
where \eqref{coprime1} is the condition for $M$ and $N$ to be right coprime and \eqref{coprime2} is the condition for $\tilde{M}$ and $\tilde{N}$ to be left coprime. By \cite[Theorem 1, p. 38]{Francis}, the internally stable controllers are given by
\begin{equation}
%\label{K}
  K =(Y-ML)(X-NL)^{-1},  \label{K1}
\end{equation}
where $L\in RH_\infty$ (i.e., the space stable proper rational matrix function) is arbitrary. This is a classical parametrization that however puts no limit on the degree of $K$, and uniqueness is not established. Choosing an $L$, we may directly obtain a sensitivity function $S(s)$.

From the above, we may write $S$ as
\begin{equation}
    S(s)=T_1(s)-T_2(s)L(s)T_3(s),
\end{equation}
where
\begin{equation}
    \label{T123}
    T_1=X\tilde{M}, \quad T_2=N,\quad T_3=\tilde{M}.
\end{equation}
Hence the (transmission) zeros of $T_2$ and $T_3$ are respectively the zeros and poles of $G(s)$. Performing inner-outer factorizations of $T_2$ and $T_3$, we obtain $T_2=\Theta_2\tilde{T_2}$ and $T_3=\tilde{T_3}\Theta_3$, where $\Theta_2$ and $\Theta_3$ are inner functions containing the unstable poles and nonminimum-phase zeros of $G$. Denote $\phi:=\det \Theta_2 \det \Theta_3$; then our interpolation points are the zeros of $\phi$ with the same multiplicities.

Now proceeding along the lines of  \cite[pages 23-25]{Francis},  we determine $X$, $\tilde{M}$ and $N$ from \eqref{G(s)}. Then inner-outer factorization of the corresponding matrices \eqref{T123} yields
%(Note that a proper real-rational function has inner-outer factorization if and only if it has no (transmission) zeros on the imaginary axis including $\infty$. Hence the approach to obtain interpolation conditions in paper \cite{BLN} is not appropriate for some cases.)

%{\color{red}
%We choose the matrices  %F=-2I, H=[-4,3;0,5]
%$$
%A_F=\begin{bmatrix}
%      -1 & 0 \\
%      0 & -2/3
%    \end{bmatrix},
%$$
%$$
%A_H=\begin{bmatrix}
%      -1 & -8/3 \\
%      0 & -37/9
%    \end{bmatrix}.
%$$
%in \cite[Lemma1, p. 23]{Francis} to be the state matrix of one minimal state space realization of $G(s)$, }
%We use normalized coprime factorization for $G(s)$,
%and by calculation we have
\begin{eqnarray}
     \Theta_2=\begin{bmatrix}
                  \frac{s-2}{s+2} & 0 \\
                  0 & \frac{5s-3}{5s+3}
                \end{bmatrix}, \quad
    \Theta_3=\begin{bmatrix}
                                          \frac{s-3}{s+3} & 0 \\
                                          0 & \frac{3s-1}{3s+1}
                                        \end{bmatrix}.
    %&& T_1=\begin{bmatrix}
     %                       \frac{ 0.83205(s-0.3333)(s+4.276)}{ (s+0.3922)^2 } & 0 \\
      %                      0 & \frac{ 0.70711 (s-7.662) (s-2) }{(s+1.581)^2}
       %                   \end{bmatrix}.
\end{eqnarray}
\begin{equation}
    \phi=\frac{(3s-1)(5s-3)(s-2)(s-3)}{(3s+1)(5s+3)(s+2)(s+3)}
\end{equation}
with interpolation points
$$
s_0=0.3333,\quad s_1=0.6,\quad s_2=2,\quad s_3=3,
$$
all of multiplicity $1$.
Define $\tilde{S}:=\phi\Theta_2^*S\Theta_3^*$ and
\begin{equation}
\begin{split}
    \tilde{T_1}&:= \phi\Theta_2^*T_1\Theta_3^*\\
    &= (s-0.3333)(s-3)\\
    &\begin{bmatrix}
    \frac{(s-0.6)(s-11)}{(s+0.3333)(s+0.6)(s+1)^2} & \frac{12(s-0.6)(s-2)}{(s+0.6)(s+1)^2(s+3)(s+4.1111)} \\
     0 & \frac{(s-2)(s+21.7778)}{(s+0.6667)(s+2)(s+3)(s+4.1111)}
      \end{bmatrix}.
\end{split}
\end{equation}
Then we have the  interpolation conditions $\tilde{S}(s_k)=\tilde{T_1}(s_k)$ for $k=0,1,2,3$ or more specifically,
\begin{equation}
\begin{split}
    \tilde{S}(s_0)=\begin{bmatrix}
                     0 & 0 \\
                     0 & 0
                   \end{bmatrix} ,~ \tilde{S}(s_1)=\begin{bmatrix}
                      0 & 0 \\
                     0 & 0.3590
                   \end{bmatrix},\\
                    \tilde{S}(s_2)=\begin{bmatrix}
                      0.3846 & 0 \\
                     0 & 0
                   \end{bmatrix},~\tilde{S}(s_3)=\begin{bmatrix}
                     0 & 0 \\
                     0 & 0
                   \end{bmatrix}.
\end{split}
\end{equation}

Denote
\begin{eqnarray}
    f(\xi):=\gamma^{-1}\tilde{S}\left(\frac{1-\xi}{1+\xi}\right)
\end{eqnarray}
with $\xi=\frac{1-s}{1+s}$. Then $f(\xi)$ is analytic outside the unit circle and hence stable. With
\begin{equation}
    z=\frac{\xi -\xi_0}{1-\xi_0\xi},
\end{equation}
we have the Carth{\'e}odory function
\begin{equation}
    \varphi(z):=\frac{1}{2}{(I-f(\xi))^{-1}}{(I+f(\xi))}
\end{equation}
with interpolation conditions
\begin{subequations}
\begin{equation}
   z_0=0,~z_1=-0.2857,~z_2= -0.7143 ,~z_3=-0.8000,
\end{equation}
\begin{equation}
\begin{split}
    \varphi(z_0)&=\tfrac{1}{2}I,~~\varphi(z_1)=\tfrac{1}{2}\begin{bmatrix}
                     1 & 0 \\
                     0 & \frac{\gamma+0.3590}{\gamma-0.3590}
                   \end{bmatrix},\\
    \varphi(z_2)&=\tfrac{1}{2}\begin{bmatrix}
                     \frac{\gamma+0.3846}{\gamma-0.3846} & 0 \\
                     0 & 1
                   \end{bmatrix},~~\varphi(z_3)=\tfrac{1}{2}I.
\end{split}
\end{equation}
\end{subequations}
where
\begin{displaymath}
R(z)=B_1(C_1B_1)^{-1}C_1z(zI-\Gamma_0)^{-1}
\end{displaymath}
However, $z(zI-\Gamma_0)^{-1}= \Gamma_0(zI-\Gamma_0)^{-1} +I$, as in \eqref{eq:decomp},
so
 \begin{displaymath}
R(z)=B_1(C_1B_1)^{-1}C_1\Gamma_0(zI-\Gamma_0)^{-1} +B_1(C_1B_1)^{-1}C_1,
\end{displaymath}

Next we choose $\gamma=10$ and apply the interpolation approach in \cite{CLtac}. Analogously with the scalar case, there is a complete parameterization of all solutions with degree constraint, and here we choose the central solution, which takes the form
\begin{equation}
    \varphi(z)=\begin{bmatrix}
                 \varphi_{11}(z) & 0 \\
                 0 & \varphi_{22}(z)
               \end{bmatrix},
\end{equation}
where
\begin{subequations}
\begin{equation}
    \varphi_{11}=-\frac{0.5(z+1.1840)(z^2+0.0842z+1.0581)}{(z-1.0121)(z^2+1.9154z+1.2379)},
\end{equation}
\begin{equation}
    \varphi_{22}=-\frac{0.5(z-1.0486)(z^2+2.2808z+1.9319)}{(z+1.5899)(z^2+0.2064z+1.2743)}.
\end{equation}
\end{subequations}
Then we go back from $\varphi(z)$ to $f(\xi)$, and then to $S(s)$. Finally we have
\begin{equation}
    Q(z)=S\left(\frac{z-1}{z+1}\right)=\begin{bmatrix}
                                              Q_{11} & 0 \\
                                              0 & Q_{22}
                                            \end{bmatrix}.
\end{equation}
where
\begin{subequations}
\begin{equation}
    Q_{11}=\frac{3.0838(z-0.25)(z+2)}{z^2-0.3283z+0.0372},
\end{equation}
\begin{equation}
    Q_{22}=\frac{1.2199(z-2)(z+0.3333)(z+0.5)}{(z-0.5)(z^2-0.6446z+0.1727)},
\end{equation}
\end{subequations}
so $Q(z)$ is stable, as required. Moreover
\begin{equation}
    H(z)=F(z)^{-1}(I-Q(z)^{-1})=\begin{bmatrix}
                                  H_{11}(z) & 0 \\
                                  0 & H_{22}(z)
                                \end{bmatrix},
\end{equation}
\begin{subequations}
where
\begin{equation}
    H_{11}=\frac{0.6757(z-0.2526)}{z+0.25},
\end{equation}
\begin{equation}
    H_{22}=\frac{0.1803(z+0.1404)(z+2.5933)}{(z+0.5)(z+0.3333)},
\end{equation}
\end{subequations}
which is also stable.

\section{Conclusion}\label{sec:conclusion}

This paper has been devoted to modeling of rank-deficient stationary vector processes, present in singular dynamic network models, dynamic factor models, etc. Any such process can be  rearranged in two  components   as in (1) to obtain two vector processes, a full rank component $u(t)$  and a residual process $y(t)$. It is shown that these components can be described by a special feedback representation  \eqref{FuHyv=0} illustrated in Figure \ref{fig_specialFB},  where $v(t)=0$, thus providing a deterministic relation between $u(t)$ and $y(t)$. In this model  the forward transfer function $F(z)$ is uniquely defined and given by  formula \eqref{Fred}. However, different choices of $u(t)$ and $y(t)$ give different $F(z)$. In general $F(z)$ is not stable, so since all processes are stationary, the complete feedback configuration needs to be internally stable, leading us to robust control theory to determine $H(z)$ and thus stabilizing the feedback loop.

\appendix
\subsection{Feedback models with uncorrelated inputs}\label{Uncorrvr}
Assume initially that the spectral density $\Phi(z)$ is full rank (invertible a.e.). Then it admits square spectral factors $W(z)$ which provide a representation of the joint process $\zeta$ as in \eqref{WNrepr} where $w(t)$ is a normalized Wiener process which can be partitioned conformably with $\zeta$ as
$$
w(t)=\bmat w_1(t)\\ w_2(t)\emat \,.
$$
Each such partitioning induces a corresponding partitioning of $W(z)$ in four blocks like that shown in the left side of equation \eqref{Wdecomp1} where an equivalent  representation in terms of the four transfer functions $F,H,G,K$ of the feedback configuration is shown. It is immediate to check that for  a full rank spectrum this representation is 1:1 and one can solve for $F,H,G,K$ in terms of $(W_{1,1},W_{1,2},W_{2,1},W_{2,2})$ uniquely, yielding a feedback representation in terms of $r(t):= K(z)w_1(t)$ and $v(t):= G(z)w_2(t)$ which are obviously uncorrelated.  Note  that in the rank-deficient case assuming $\dim w_1(t)=m$ and hence $\dim w_2(t)=0$ the first of the two relations \eqref{eq:W2FH} and that for the $m\times m$ noise matrix  $K(z)$ still hold true  while that for $H(z)$  in the rank-deficient case must involve a generalized inverse  and hence implies  non-uniqueness.  Of course in this case  $G(z)$ is chosen equal to zero.

\subsection{Proof of Theorem \ref{MainThm}}\label{appendixThm_2}
\begin{proof}
The claim is equivalent to the following two statements:\\
1. (Sufficiency) If we have the particular feedback structure \eqref{FuHyv=0} i.e. $\Phi_v \equiv 0$;  then $u$ is   of full rank $m= \rank (\Phi)$.\\
2. (Necessity) Conversely, if $u$ is   of full rank $m= \rank (\Phi)$ then $\Phi_v \equiv 0$.

Part 1 follows  from \eqref{Phivr} since   then $\Phi_r$ must have rank $m = \rank (\Phi)$.\\
Part 2 is not so immediate. It  is proved as follows.

Since 	$\Phi(z)$ has rank $m$ a.e. there must be a full rank $(p-m)\times p$ rational matrix which we write in partitioned form, such that
\begin{align*}
\![A(z)~B(z)] \Phi(z) &=0\, \!\Leftrightarrow \\
 [A(z)~B(z)] & \, \E \{\bmat\hat u(z)\\ \hat y(z)\emat \bmat \hat u(z)^* & \hat y(z)^*\emat \} =0  \\
&\, \Leftrightarrow \, [A(z)~B(z)] \bmat u(t) \\ y(t)\emat \perp \Hb(u,y)
\end{align*}
where $A$, $B$ are $(p-m)\times m$, $(p-m)\times (p-m)$ matrices and hats denote Fourier transforms \cite{LPbook}. The last formula implies that
\begin{equation}\label{AB}
 A(z)u(t) + B(z) y(t) \equiv 0
\end{equation}
 identically, since the term on the left  is a  member of the Hilbert space $\Hb(u,y) $.
We claim that $B(z)$ must be of full rank $p-m$. One can prove this using the invertibility of $\Phi_{u}(z)$ as follows.
For, suppose $B(z)$ is singular, then  pick a $p-m$-dimensional non-zero row vector $a(z)$ in the left null space of $B(z)$ and multiply from the left the last  relation by  $a(z)$. This would imply that also $a(z)A(z)\Phi_{u}(z)=0$ which in turn implies $a(z)A(z)=0$ since $\Phi_{u}$ is full rank. However $a(z)[A(z)~B(z)]$ cannot be zero for the  matrix $[A(z)~B(z)]$ is full rank $p-m$ and hence $a(z)$ must be zero. So $B(z)$ must be full rank.\\
Now  take the nonsingular $(p-m)\times (p-m)$ rational matrix $M(z)=B(z)^{-1}$ and consider instead $M(z)[A(z)~B(z)]$, which  provides an equivalent relation to \eqref{AB}. By this
 we can reduce $B(z)$ to the identity getting a relation of the type
$$
[ -F(z)~\, I\,] \bmat  u(t) \\ y(t)\emat =0
$$
where  $F(z)$ is a rational matrix function, so that    one gets the deterministic dynamical relation
$$
  y(t)= F(z) u(t)\,.
$$
Substituting from  the general feedback model the relation $u(t)=H(z)y(t)+ r$ one concludes  that $y(t)= Q(z) r(t)$ must thens be a functional of only the noise $r$ and likewise $u(t)$. Therefore by \eqref{Phivr} and the uncorrelation of $v$ and $r$ one must conclude that  $\Phi_v=0$, i.e. $v$ must be  the zero process. Hence a representation like \eqref{FuHyv=0}  must hold.
\end{proof}

\subsection{Some lemmas}\label{appendixlemmas}

\begin{lemma}\label{inverse}
Let $G(z)$ be the proper rational transfer function
\begin{displaymath}
G(z)=C(zI-A)^{-1}B +D
\end{displaymath}
where $D$ is square and nonsingular. Then
\begin{displaymath}
G(z)^{-1}=D^{-1} - D^{-1}C\left[zI-(A-BD^{-1}C)\right]^{-1}BD^{-1}.
\end{displaymath}
\end{lemma}

\medskip

\begin{proof}
The rational matrix function $G(z)$ is the transfer function of the control system
\begin{align*}
 x(t+1)   &=Ax(t)+Bu(t)   \\
  y(t)  &  =Cx(t)+Du(t)
\end{align*}
Inserting $u(t)=D^{-1}\left[y(t)-Cx(t)\right]$ in the first equation yields the inverse system
\begin{align*}
  x(t+1)  & =(A-BD^{-1}C)x(t)+BD^{-1}y(t)  \\
  u(t)  &  =-D^{-1}Cx(t) + D^{-1}y(t)
\end{align*}
with transfer function $G(z)^{-1}$.
\end{proof}

\medskip

\begin{lemma}\label{similaritylem}
    If $A\in \mathbb{R}^{m\times n}$, $B\in \mathbb{R}^{n\times m}$, then the nonzero eigenvalues of $AB$ and $BA$ are the same.
\end{lemma}

     \medskip

\begin{proof}
The two matrices
\begin{displaymath}
T_1:=\begin{bmatrix}AB & 0\\ B & 0\end{bmatrix}\quad\text{and}\quad
T_2:=\begin{bmatrix}0 & 0\\ B & BA\end{bmatrix}
\end{displaymath}
 are similar. In fact,  $S^{-1}T_1S =T_2$,
where  $S$ is the $(m+n)\times (m+n)$ matrix
\begin{displaymath}
S:= \begin{bmatrix} I_m & A\\ 0 & I_n \end{bmatrix}.
\end{displaymath}
Therefore $T_1$ and $T_2$ have the same characteristic polynomials, i.e.,
\begin{displaymath}
\det(\lambda I_m - AB)\lambda^n=\lambda^m\det(\lambda I_n - BA).
\end{displaymath}
Therefore each nonzero eigenvalue $\tilde\lambda$ of $AB$ is also an eigenvalue of $BA$ and vice versa.
 \end{proof}

 \subsection{Determining $H(z)$ in the scalar case}\label{appendixH}
Given a scalar transfer function $F(z)$ with an unstable pole in $z=\xi_0^{-1}$ and a nonminimum-phase zero in $z=\xi_1^{-1}$, we want to determine a scalar $H(z)$ so that the feedback loop in Figure~\ref{fbdtfig} is internally stable. As explained in Section~\ref{sec:H}, this amounts to determining a Carath{\'e}odory function $\varphi$ which satisfies the interpolation conditions $\varphi(0)=\tfrac12$ and $\varphi(z_1)=w_1$, where $z_1=(\xi_1-\xi_0)(1-\xi_0\xi_1)^{-1}$ and $w_1=\tfrac12(\gamma-1)(\gamma+1)^{-1}$.

 To this end, we shall use the Riccati-type approach of \cite{CLtac}. In the problem formulation in the introduction of that paper $m=1$, $n_0=n_1=1$ and $n=1$, and from Section III-B in the same paper, we have the matrices
 \begin{displaymath}
W=\begin{bmatrix}\tfrac12&0\\0&w_1\end{bmatrix}, \; Z=\begin{bmatrix}0&0\\0&z_1\end{bmatrix}, \; e=\begin{bmatrix}1\\1\end{bmatrix}, \;
V=\begin{bmatrix}1&0\\1&z_1\end{bmatrix},
\end{displaymath}
which yields
\begin{displaymath}
T=\begin{bmatrix}0&0\\0&-\gamma^{-1}\end{bmatrix}
\end{displaymath}
and thus, continuing along the lines of \cite[Section III-B]{CLtac},  we have the interpolation data
\begin{equation}
\label{uU}
u=-\gamma^{-1}z_1^{-1}, \quad U=-\gamma^{-1}.
\end{equation}
With $\sigma$ an arbitrary parameter in the interval $(-1,1)$, the  Riccati-type equation then becomes
\begin{equation}
\label{CEE}
p=\sigma^2(p-p^2)+(u+U\sigma-U\sigma p)^2,
\end{equation}
which, by \cite[Theorem 9]{CLtac}, has a unique solution $0<p<1$. Then the corresponding solution of the interpolation problem is
\begin{displaymath}
\varphi(z)=\tfrac12\frac{1+bz}{1+az},
\end{displaymath}
where
\begin{displaymath}
\begin{split}
a&=(1-U)\sigma (1-p)-u\\
b&=(1+U)\sigma (1-p)+u.
\end{split}
\end{displaymath}
The central solution is obtained by setting $\sigma=0$, yielding
\begin{equation}
\label{phicentral}
\varphi(z)=\tfrac12\frac{1+uz}{1-uz}.
\end{equation}
To obtain a general solution for a nonzero $\sigma$ we need to solve the nonlinear equation \eqref{CEE}, which can be done by the  homotopy continuation method in subsection III-E of \cite{CLtac}.
The parameter $\gamma$ has to be chosen so that the Pick condition in \cite[Proposition 3]{CLtac} is satisfied.

\subsection{Proof of Theorem~\ref{thm:ufullrank}}\label{appendixthm3}

We need to show that the $m\times m$ spectral factor
\begin{displaymath}
\begin{bmatrix}
W_{00}(z)&W_{01}(z)\\W_{10}(z)&W_{11}(z)
\end{bmatrix}
\end{displaymath}
is full rank or, equivalently, that the Schur complement
\begin{equation}
\label{eq:Schur}
S(z)=W_{11}(z)- W_{10}(z)W_{00}(z)^{-1}W_{01}(z)
\end{equation}
is full rank. To this end, we first form
\begin{equation}
\label{eq:W10W00inv}
W_{10}(z)W_{00}(z)^{-1}= C_1(zI-A)^{-1}Q(z)B_0\Sigma^{-1},
\end{equation}
where
\begin{align*}
Q(z)&=I-B_0\Sigma^{-1}C_0(zI-\Gamma_0)^{-1} \\
    & = \left[zI-\Gamma_0-B_0\Sigma^{-1} C_0\right](zI-\Gamma_0)^{-1} \\
    &=(zI-A)(zI-\Gamma_0)^{-1},
\end{align*}
which inserted into \eqref{eq:W10W00inv} yields
\begin{equation}
\label{eq:W10W00inv2}
W_{10}(z)W_{00}(z)^{-1}=C_1(zI-\Gamma_0)^{-1}B_0\Sigma^{-1},
\end{equation}
and hence
\begin{align*}
&W_{10}(z)W_{00}(z)^{-1}W_{01}\\&= C_1(zI-\Gamma_0)^{-1}B_0\Sigma^{-1} C_0(zI-A)^{-1}B_1\\
&=C_1(zI-A)^{-1}B_1 - C_1(zI-\Gamma_0)^{-1}B_1\\
&=W_{11}(z) - C_1(zI-\Gamma_0)^{-1}B_1,
\end{align*}
where we have used the fact that
\begin{displaymath}
B_0\Sigma^{-1} C_0=A-\Gamma_0= (zI-\Gamma_0) -(zI-A)^{-1}.
\end{displaymath}
Consequently, the Schur complement \eqref{eq:Schur} is given by
\begin{equation}
\label{eq:Schur1}
S(z)=C_1(zI-\Gamma_0)^{-1}B_1.
\end{equation}
To see that $S(z)$ is full rank, first note thet
\begin{align*}
(zI-\Gamma_0)^{-1}= & z^{-1}(zI-\Gamma_0+\Gamma_0)(zI-\Gamma_0)^{-1} \\
 = & z^{-1}I +z^{-1}\Gamma_0(zI-\Gamma_0)^{-1}.
\end{align*}
to obtain
\begin{equation}
\label{eq:Schur2}
S(z)=z^{-1}\left[C_1B_1+C_1\Gamma_0(zI-\Gamma_0)^{-1}B_1\right],
\end{equation}
which is clearly full rank whenever $C_1B_1$ is nonsingular.

\subsection{Proof of Theorem~\ref{thm:F(z)}}\label{appendixthm4}
Since $u$ is full rank (Theorem~\ref{thm:ufullrank}), $y=F(z)u$ is given by
\begin{displaymath}
y=\begin{bmatrix}W_{20}&W_{21}\end{bmatrix}\begin{bmatrix}
W_{00}&W_{01}\\W_{10}&W_{11}\end{bmatrix}^{-1}
\begin{bmatrix}u_0\\u_1\end{bmatrix},
\end{displaymath}
and therefore
\begin{align*}
F_0W_{00}+F_1W_{10}&=W_{20}\\
F_0W_{01}+F_1W_{11}&=W_{21},
\end{align*}
from which we have
\begin{subequations}
\begin{align}
   F_1(z) &=T(z)S(z)^{-1} \label{eq:ST2F1}  \\
  F_0(z) & =  W_{20}(z)W_{00}(z)^{-1} -F_1(z)W_{10}(z)W_{00}(z)^{-1}\label{eq:F1prel}
\end{align}
\end{subequations}
where $S(z)$ is given by \eqref{eq:Schur1} or \eqref{eq:Schur2} and
\begin{equation}
\label{eq:T}
T(z)=W_{21}(z)-W_{20}W_{00}(z)^{-1}W_{01}(z),
\end{equation}
which clearly is obtained by exchanging $C_1$ by $C_2$ in the calculation leading to \eqref{eq:Schur1}. Consequently,
\begin{equation}
\label{eq:T(z)}
T(z)=C_2(zI-\Gamma_0)^{-1}B_1.
\end{equation}
To determine $F_1(z)$ we apply Lemma~\ref{inverse} in the appendix to obtain
 \begin{align*}
  z^{-1}&B_1S(z)^{-1}     \\
    & =  \left[I- B_1(C_1B_1)^{-1}C_1\Gamma_0(zI-\Gamma_1)^{-1}\right]B_1(C_1B_1)^{-1},
\end{align*}
 where $\Gamma_1$ is given by \eqref{eq:Gamma1}. However,
 \begin{equation}
 \label{eq:Gammao-Gamma1}
 \begin{split}
B_1(C_1B_1)^{-1}C_1\Gamma_0&\\ = \Gamma_0-\Gamma_1&=(zI-\Gamma_1)-(zI-\Gamma_0)
\end{split}
\end{equation}
and therefore
\begin{displaymath}
 z^{-1}B_1S(z)^{-1}=(zI-\Gamma_0)(zI-\Gamma_1)^{-1}B_1(C_1B_1)^{-1},
\end{displaymath}
which together with \eqref{eq:ST2F1} and \eqref{eq:T(z)} yields \eqref{eq:F1}. To derive \eqref{eq:F1alt} just insert
\begin{equation}
\label{eq:decomp}
z(zI-\Gamma_1)^{-1}= \Gamma_1(zI-\Gamma_1)^{-1} +I
\end{equation}
into \eqref{eq:F1}.

To determine $F_0$ from \eqref{eq:F1prel} we first note that a calculation analogous to that leading to \eqref{eq:W10W00inv} yields
\begin{equation}
\label{eq:W20W00inv}
W_{20}(z)W_{00}(z)^{-1}=C_2(zI-\Gamma_0)^{-1}B_0\Sigma^{-1}.
\end{equation}
Moreover, from \eqref{eq:F1} and \eqref{eq:W10W00inv} we obtain
\begin{equation}
\label{eq:F1W10W00inv}
F_1(z)W_{10}(z)W_{00}(z)^{-1}=C_2(zI-\Gamma_1)^{-1}R(z)B_0\Sigma^{-1},
\end{equation}
where
\begin{displaymath}
R(z)=B_1(C_1B_1)^{-1}C_1z(zI-\Gamma_0)^{-1}
\end{displaymath}
However, $z(zI-\Gamma_0)^{-1}= \Gamma_0(zI-\Gamma_0)^{-1} +I$, as in \eqref{eq:decomp},
so
 \begin{displaymath}
R(z)=B_1(C_1B_1)^{-1}C_1\Gamma_0(zI-\Gamma_0)^{-1} +B_1(C_1B_1)^{-1}C_1,
\end{displaymath}
which together with \eqref{eq:Gammao-Gamma1} yields
\begin{displaymath}
R(z)=(zI-\Gamma_1)(zI-\Gamma_0)^{-1}-I+B_1(C_1B_1)^{-1}C_1.
\end{displaymath}
Inserting this expression for $R(z)$ into \eqref{eq:F1W10W00inv} we have
\begin{equation}
\label{eq:F0secondpart}
\begin{split}
&F_1(z)W_{10}(z)W_{00}(z)^{-1}=C_2(zI-\Gamma_0)^{-1}B_0\Sigma^{-1} \\
&\quad- C_2(zI-\Gamma_1)^{-1}\left[I-B_1(C_1B_1)^{-1}C_1\right]B_0\Sigma^{-1},  \\
\end{split}
\end{equation}
Then \eqref{eq:F0} follows directly from \eqref{eq:F1prel}, \eqref{eq:W20W00inv} and \eqref{eq:F0secondpart}.

\begin{IEEEbiography}[{\includegraphics[width=1in,height=1.25in,clip,keepaspectratio]{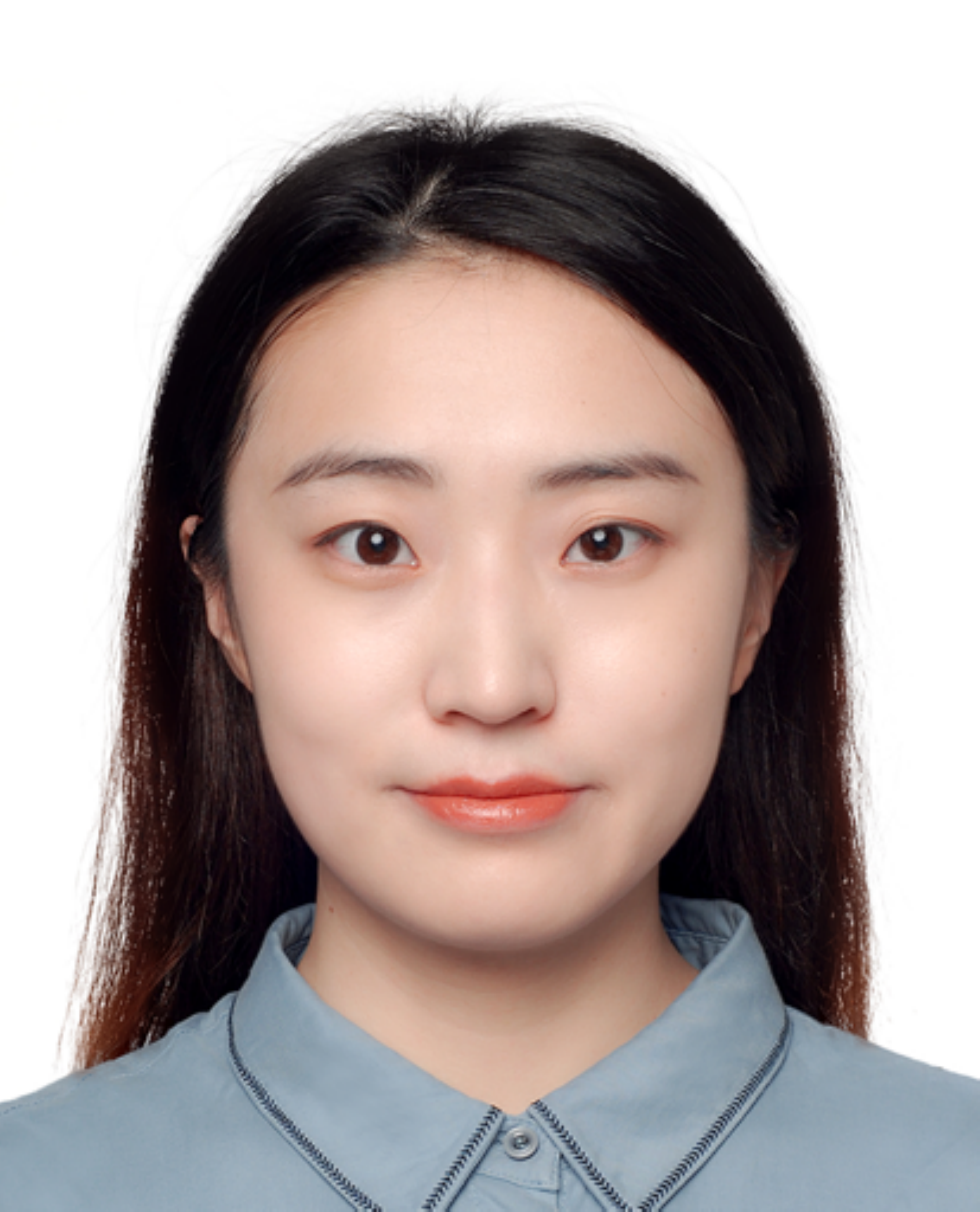}}]{Wenqi Cao} (S'20) received the B.E. degree in intelligent science and technology from Nankai University, Tianjin, China, in 2018. She is now a Ph.D student in Shanghai Jiao Tong University and a visiting Ph.D student in University of Padova, Padova, Italy. Her current research insterest is control and identification of low rank linear stochastic systems.
Wenqi Cao is a reviewer of  IFAC Symposium on System Identification (SYSID) 2021, and of Chinese Automation Congress (CAC) 2021. She was the winner of  Frontrunner 5000 Top Articles in Outstanding S\&T Journals of China in 2019.
\end{IEEEbiography}

\begin{IEEEbiography}[{\includegraphics[width=1in,height=1.25in,clip,keepaspectratio]{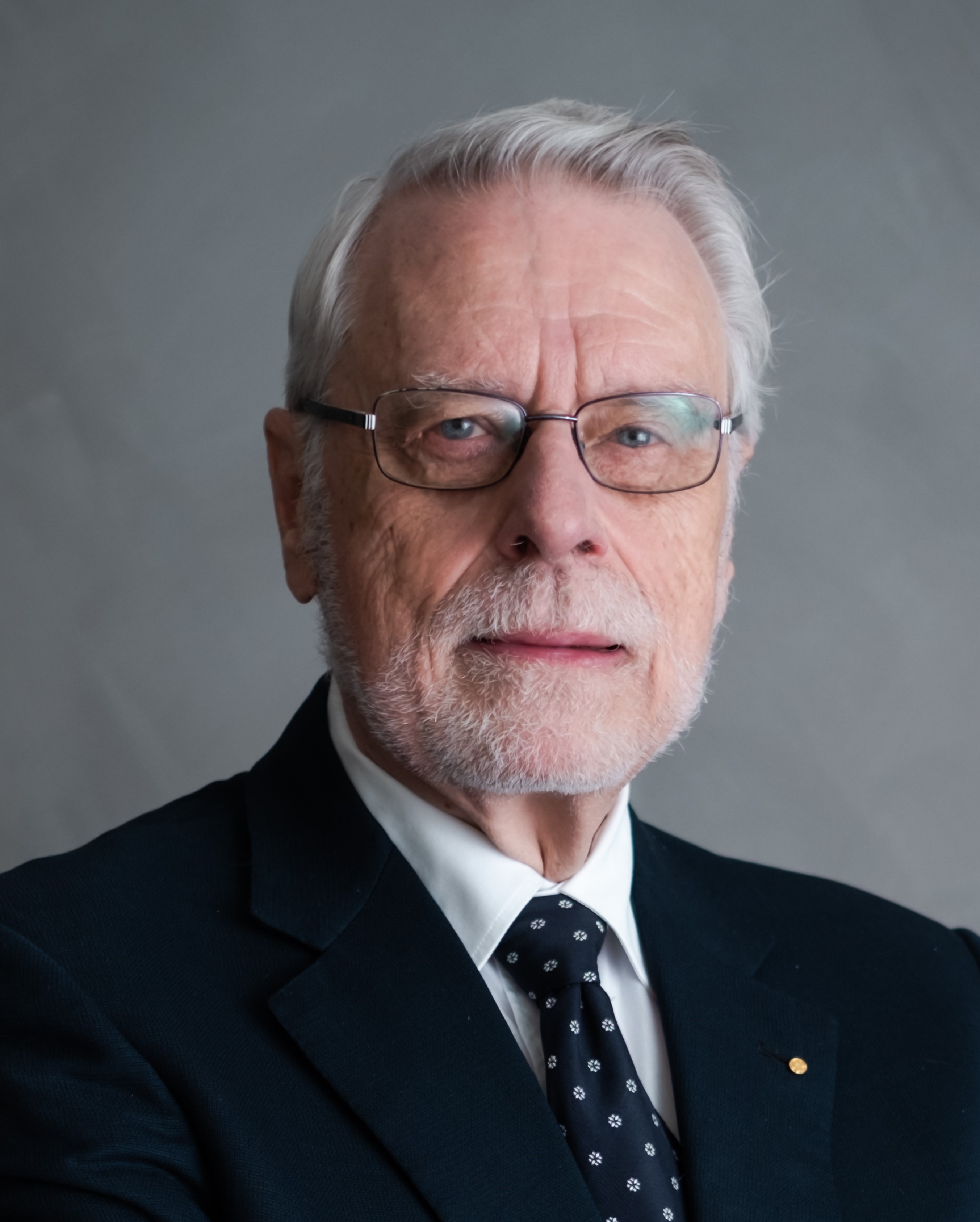}}]{Anders Lindquist}
(M’77--SM’86--F’89--LF’10) received the Ph.D. degree in optimization and systems theory from the Royal Institute of Technology, Stockholm, Sweden, in 1972, and an honorary doctorate (Doctor Scientiarum Honoris Causa) from Technion (Israel Institute of Technology) in 2010.

He is currently a Zhiyuan Chair Professor at Shanghai Jiao Tong University,  China, and Professor Emeritus at the Royal Institute of Technology (KTH), Stockholm, Sweden.  Before that he had a full academic career in the United States, after which he was appointed to the Chair of Optimization and Systems at KTH.

Dr. Lindquist is a Member of the Royal Swedish Academy of Engineering Sciences, a Foreign Member of the Chinese Academy of Sciences, a Foreign Member of the Russian Academy of Natural Sciences (elected 1997), a Member of Academia Europaea (Academy of Europe), an Honorary Member the Hungarian Operations Research Society, a Fellow of SIAM, and a Fellow of IFAC. He received the 2003 George S. Axelby Outstanding Paper Award, the 2009 Reid Prize in Mathematics from SIAM, and the 2020 IEEE Control Systems Award, the IEEE field award in Systems and Control.
\end{IEEEbiography}

\begin{IEEEbiography}[{\includegraphics[width=1in,height=1.25in,clip,keepaspectratio]{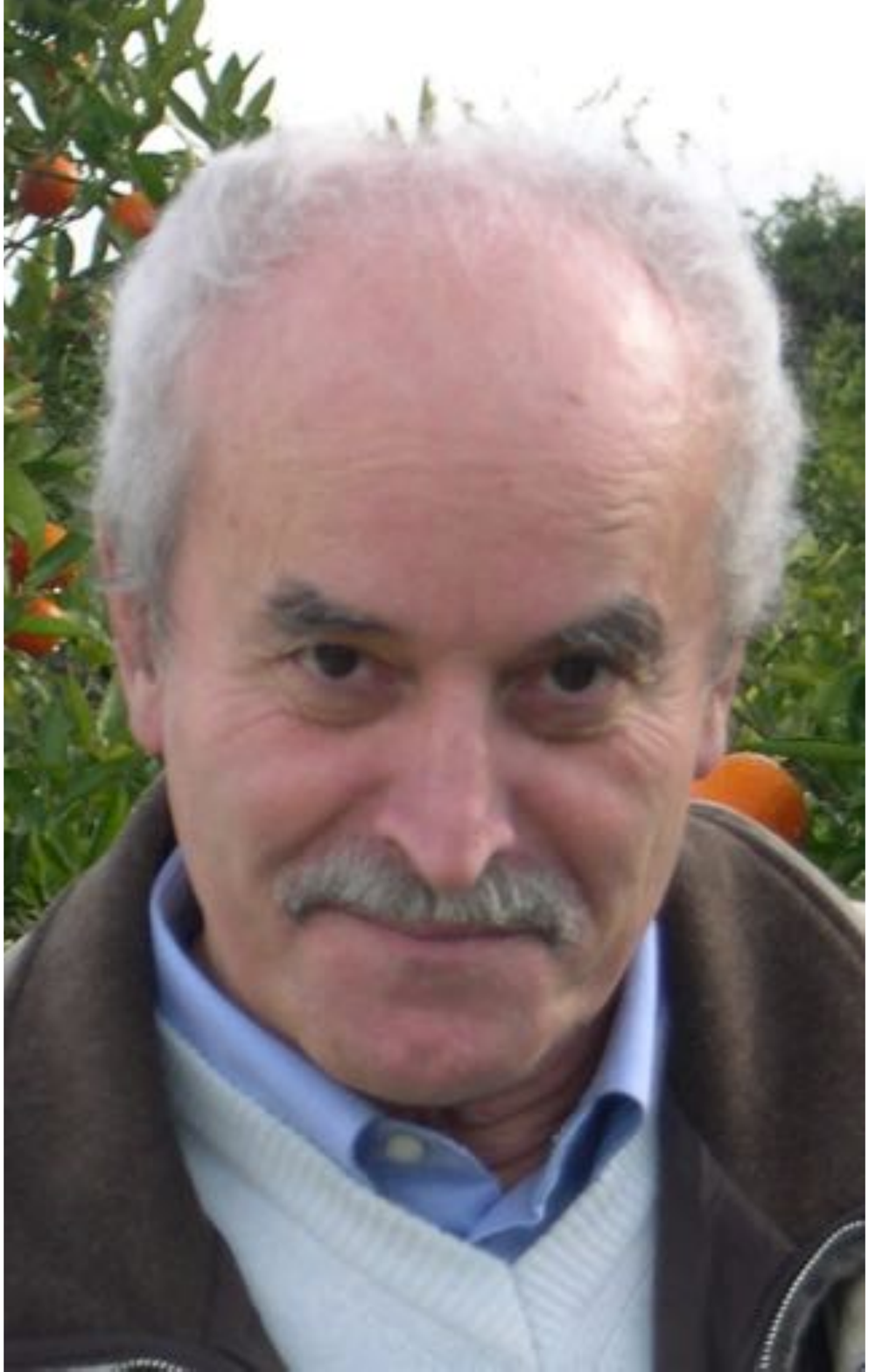}}] {Giorgio Picci}
(S’67--M’70--SM’91--F’94--LF’08) received the Dr.Eng. degree from the University of Padua, Padua, Italy, in 1967.
Currently, he is Professor Emeritus with the Department of Information Engineering, University of Padua, Padua, Italy. He has held several long-term visiting appointments with various American, Japanese, and European universities among which Brown University, MIT, the University of Kentucky, Arizona State University, the Center for Mathematics and Computer Sciences (CWI) in Amsterdam, the Royal Institute of Technology, Stockholm, Sweden, Kyoto University, and Washington University, St. Louis, Mo, USA. He has been contributing to systems and control mostly in the area of modeling, estimation, and identification of stochastic systems and published over 150 papers and written or edited several books in this area. He has been involved in various joint research projects with industry and state agencies. Besides being a life Fellow of IEEE, he
is a Fellow of IFAC and a foreign member of the Swedish Royal Academy of Engineering Sciences.
\end{IEEEbiography}

\end{document}